\documentclass[11pt,reqno]{article}

\usepackage{amssymb,amscd,amsmath,amsthm}
\numberwithin{equation}{section}
\usepackage{graphicx}
\newtheorem{theorem}{Theorem}[section]
\newtheorem{proposition}[theorem]{Proposition}
\newtheorem{lemma}[theorem]{Lemma}
\newtheorem{corollary}[theorem]{Corollary}
\newtheorem{definition}[theorem]{Definition}
\newtheorem{remark}[theorem]{Remark}

\def\cb{{\mathcal B}}

\def\ce{{\mathcal E}}

\def\cw{{\mathcal W}}

\def\bc{{\mathbb C}}
\def\bh{{\mathbb H}}

\def\bn{{\mathbb N}}

\def\br{{\mathbb R}}

\def\frak{\mathfrak}

\def\ga{{\frak A}}

\def\a{\alpha}
\def\b{\beta}

\def\tr{{\rm Tr}}
\def\L{\Lambda}
\def\G{\Gamma}
\def\ce{\mathcal E}

\def\ffi{\varphi}

\def\Tr{\mathrm{Tr}}
\def\<{\langle}
\def\>{\rangle}

\def\1{\mathbf{1}}

\def\cw{\cal W}

\def\cal{\mathcal}

\def\s{\sigma}
\def\bh{\mathbf{h}}
\def\bs{\mathbf{s}}
\def\bg{\mathbf{g}}

\def\id{{\bf 1}\!\!{\rm I}}

\begin{document}

\begin{center}
{\Large {\bf On Quantum Markov Chains on Cayley tree III:\\
 Ising model}}\\[1cm]
\end{center}

\begin{center}
{\large {\sc Luigi Accardi}}\\[2mm]
\textit{Centro Interdisciplinare Vito Volterra\\
II Universit\`{a} di Roma ``Tor Vergata''\\
Via Columbia 2, 00133 Roma, Italy} \\
E-email: {\tt accardi@volterra.uniroma2.it}
\end{center}

\begin{center}
{\large {\sc Farrukh Mukhamedov}}\\[2mm]
\textit{ Department of Computational \& Theoretical Sciences,\\
Faculty of Science, International Islamic University Malaysia,\\
P.O. Box, 141, 25710, Kuantan, Pahang, Malaysia}\\
E-mail: {\tt far75m@yandex.ru, \ farrukh\_m@iiu.edu.my}
\end{center}

\begin{center}
{\large{\sc Mansoor Saburov}}\\[2mm]
\textit{Department of Computational \& Theoretical Sciences,\\
Faculty of Science, International Islamic University Malaysia,\\
P.O. Box, 141, 25710, Kuantan, Pahang, Malaysia}\\
E-mail: {\tt msaburov@gmail.com}
\end{center}

\begin{abstract}
In this paper, we consider the classical Ising model on the Cayley
tree of order $k$ ($k\geq 2$), and show the existence of the phase
transition in the following sense: there exists two quantum Markov
states which are not quasi-equivalent. It turns out that the found
critical temperature coincides with usual critical temperature.
\end{abstract}

\section{Introduction }\label{intr}

The present paper is a continuation of our previous works
\cite{AMSa,AMSaII}. In \cite{AMSa} we have introduced forward type
of quantum Markov chains (QMC) defined on the Cayley tree was
studied\footnote{We remark that backward quantum Markov chains on
lattices and trees have been investigated in \cite{[AcFi01a],AOM}.}.
It was provided a construction of such kind of chains, in which a
QMC is defined as a weak limit of finite volume states with boundary
conditions\footnote{Note that similar kind of constructions of QMC
on integer lattice were known in the literature (see for example
\cite{Ac87}-\cite{AcLi}.}. By means of the provided construction we
proved uniqueness of QMC associated with $XY$-model on a Cayley tree
of order two. Furthermore, in \cite{AMSaII} we have defined a notion
of phase transition in QMC scheme. Namely, such a notion is based on
the quasi-equivalence of QMC. Therefore, such a phase transition is
purely noncommutative. We point out that phase transitions in a
quantum setting play an important role to understand quantum spin
systems (see for example \cite{ArE,BCS,FS}). Furthermore, in
\cite{AMSaII} it was established the existence of the phase
transition for $XY$-model on the Cayley tree of order three. From
the provided definition of the phase transition, it naturally
appears a question: would this definition be compatible with
well-known definitions of phase transitions for lattice models (see
\cite{Bax,Geor,Roz2}).

In this paper, we consider the classical Ising model on the Cayley
tree of order $k$ ($k\geq 2$), and show the existence of the phase
transition in sense of \cite{AMSaII}. It turns out that the found
critical temperature coincides with usual one (see
\cite{Geor,Roz2}). This means that the our definition of the phase
transition is compatible with known ones. Note that very recently,
in \cite{Gond} other new kind of phase transitions have been
observed, for the classical Ising model on Caley tree. We stress
that noncommutative approach to the phase transition for the Ising
model was studied in \cite{ArE}. But our way to define the phase
transition is bit different from the mentioned paper. In general,
phase transition for quantum systems is defined as the existence of
two distinct KMS-state corresponding to the model (see \cite{BR2}
for review). In our approach we follow the same definition but we
require the existence of two non-quasi equivalent QMC associated
with a model. We point out that, in general, two distinct QMC may
generate in GNS-representation the same type of von Neumann algebras
\cite{MR2}.

The paper is organized as follows. In section 2 we give necessary
definitions and construction of QMC. In section 3 we consider Ising
model and formulate the main result of the paper. In section 4 we
derive a dynamical system related to our model. In section 5
asymptotic behavior of the dynamical system will be studied. In
section 6 we prove the diaganalizability of the QMC. In section 7
the first part of the main theorem will be proved. In the final
section 8, it will be established the existence of the phase
transition.

\section{Construction of QMC on the Cayley tree}\label{dfcayley}

In this section we recall needed definitions which will be used in
the paper (see \cite{AMSa,AMSaII} for more information).

Let $\Gamma^k = (L,E)$ be a semi-infinite Cayley tree of order
$k\geq 1$ with the root $x^0$ (i.e. each vertex of $\Gamma^k$ has
exactly $k+1$ edges, except for the root $x^0$, which has $k$
edges). Here $L$ is the set of vertices and $E$ is the set of edges.
If the vertices $x$ and $y$ are connected by an edge, then they are
called {\it nearest neighbors} and denoted by $<x,y>$. A collection
of the pairs $<x,x_1>,\dots,<x_{d-1},y>$ is called a {\it path} from
the point $x$ to the point $y$. The distance $d(x,y), x,y\in V$, on
the Cayley tree, is the length of the shortest path from $x$ to $y$.
One can define a coordinate structure in $\G^k$ as follows:  every
vertex $x$ (except for $x^0$) of $\G^k$ has coordinates
$(i_1,\dots,i_n)$, here $i_m\in\{1,\dots,k\}$, $1\leq m\leq n$ and
for the vertex $x^0$ we put $(0)$.  Namely, the symbol $(0)$
constitutes level 0, and the sites $(i_1,\dots,i_n)$ form level $n$
( i.e. $d(x^0,x)=n$) of the lattice (see Fig. 1).


Let us set
\[
W_n = \{ x\in L \, : \, d(x,x_0) = n\} , \qquad \Lambda_n =
\bigcup_{k=0}^n W_k, \qquad  \L_{[n,m]}=\bigcup_{k=n}^mW_k, \ (n<m)
\]
\[
E_n = \big\{ <x,y> \in E \, : \, x,y \in \Lambda_n\big\}, \qquad
\Lambda_n^c = \bigcup_{k=n}^\infty W_k
\]

Now us rewrite the elements of $W_n$ in the following order, i.e.
\begin{eqnarray*}
\overrightarrow{W_n}:=\left(x^{(1)}_{W_n},x^{(2)}_{W_n},\cdots,x^{(|W_n|)}_{W_n}\right),\quad
\overleftarrow{W_n}:=\left(x^{(|W_n|)}_{W_n},x^{(|W_n|-1)}_{W_n},\cdots,
x^{(1)}_{W_n}\right).
\end{eqnarray*}
Note that $|W_n|=k^n$. Vertices
$x^{(1)}_{W_n},x^{(2)}_{W_n},\cdots,x^{(|W_n|)}_{W_n}$ of $W_n$ can
be represented in terms of the coordinate system as follows
\begin{eqnarray*}
&&x^{(1)}_{W_n}=(1,1,\cdots,1,1), \quad x^{(2)}_{W_n}=(1,1,\cdots,1,2), \ \ \cdots \quad x^{(k)}_{W_n}=(1,1,\cdots,1,k,),\\
&&x^{(k+1)}_{W_n}=(1,1,\cdots,2,1), \quad
x^{(2)}_{W_n}=(1,1,\cdots,2,2), \ \ \cdots \quad
x^{(2k)}_{W_n}=(1,1,\cdots,2,k),
\end{eqnarray*}
\[\vdots\]
\begin{eqnarray*}
&&x^{(|W_n|-k+1)}_{W_n}=(k,k,,\cdots,k,1), \
x^{(|W_n|-k+2)}_{W_n}=(k,k,\cdots,k,2),\ \ \cdots
x^{|W_n|}_{W_n}=(k,k,\cdots,k,k).
\end{eqnarray*}

For a given vertex $x,$ we shall use the following notation for the
set of \textit{direct successors} of $x$:
\begin{eqnarray*}
\overrightarrow{S(x)}:=\left((x,1),(x,2),\cdots (x,k)\right),\quad
\overleftarrow{S(x)}:=\left((x,k),(x,k-1),\cdots (x,1)\right).
\end{eqnarray*}

In what follows, for the sake of simplicity, we will use notation
$i\in \overrightarrow{S(x)}$ (resp. $i\in \overleftarrow{S(x)}$
instead of $(x,i)\in \overrightarrow{S(x)}$ (resp. $(x,i)\in
\overleftarrow{S(x)}$).

The algebra of observables $\cb_x$ for any single site $x\in L$ will
be taken as the algebra $M_d$ of the complex $d\times d$ matrices.
The algebra of observables localized in the finite volume $\L\subset
L$ is then given by $\cb_\L=\bigotimes\limits_{x\in\L}\cb_x$. As
usual if $\L^1\subset\L^2\subset L$, then $\cb_{\L^1}$ is identified
as a subalgebra of $\cb_{\L^2}$ by tensoring with unit matrices on
the sites $x\in\L^2\setminus\L^1$. Note that, in the sequel, by
$\cb_{\L,+}$ we denote the positive part of $\cb_\L$. The full
algebra $\cb_L$ of the tree is obtained in the usual manner by an
inductive limit
$$
\cb_L=\overline{\bigcup\limits_{\L_n}\cb_{\L_n}}.
$$


Assume that for each edge $<x,y>\in E$ of the tree an operator
$K_{<x,y>}\in {\cal B}_{\{x,y\}}$ is assigned. We would like to
define a state on $\cb_{\L_n}$ with boundary conditions $w_{0}\in
{\cal B}_{(0),+}$ and $\bh=\{h_x\in {\cal B}_{x,+}\}_{x\in L}$.

Let us denote
\begin{eqnarray}
K_{[m-1,m]}&:=&\prod_{x\in
\overrightarrow{W}_{m-1}}\prod_{y\in \overrightarrow{S(x)}}K_{<x,y>},\\
\bh^{1/2}_n&:=&\prod_{x\in \overrightarrow{W}_n}h_x^{1/2}, \quad \quad \bh_n:=\bh^{1/2}_n(\bh^{1/2}_n)^{*},\\
K_n&:=&w_0^{1/2}K_{[0,1]}K_{[1,2]}\cdots K_{[n-1,n]}\bh^{1/2}_n,\\
\label{w_n}{\cw}_{n]}&:=&K_nK_n^{*},
\end{eqnarray}
It is clear that ${\cw}_{n]}$ is positive.

In what follows, by $\tr_{\L}:\cb_L\to\cb_{\L}$ we mean normalized
partial trace (i.e. $\tr_{\L}(\id_{L})=\id_{\L}$, here
$\id_{\L}=\bigotimes\limits_{y\in \L}\id$), for any
$\Lambda\subseteq_{\text{fin}}L$. For the sake of shortness we put
$\tr_{n]} := \tr_{\Lambda_n}$.

Let us define a positive functional $\ffi^{(n,f)}_{w_0,\bh}$ on
$\cb_{\Lambda_n}$ by
\begin{eqnarray}\label{ffi-ff}
\ffi^{(n,f)}_{w_0,\bh}(a)=\tr(\cw_{n+1]}(a\otimes\id_{W_{n+1}})),
\end{eqnarray}
for every $a\in \cb_{\Lambda_n}$. Note that here, $\tr$ is a
normalized trace on ${\cal B}_L$ (i.e. $\tr(\id_L)=1$).

To get an infinite-volume state $\ffi^{(f)}$ on $\cb_L$  such
that $\ffi^{(f)}\lceil_{\cb_{\L_n}}=\ffi^{(n,f)}_{w_0,\bh}$, we
need to impose some constrains to the boundary conditions
$\big\{w_0,\bh\big\}$ so that the functionals
$\{\ffi^{(n,f)}_{w_0,\bh}\}$ satisfy the compatibility condition,
i.e.
\begin{eqnarray}\label{compatibility}
\ffi^{(n+1,f)}_{w_0,\bh}\lceil_{\cb_{\L_n}}=\ffi^{(n,f)}_{w_0,\bh}.
\end{eqnarray}

\begin{theorem}[\cite{AMSa}]\label{compa} Let the boundary conditions $w_{0}\in {\cal
B}_{(0),+}$ and ${\bh}=\{h_x\in {\cal B}_{x,+}\}_{x\in L}$ satisfy
the following conditions:
\begin{eqnarray}\label{eq1}
&& \tr ( w_0 h_0 ) =1 \\
\label{eq2}
&&\tr_{x]}\left[\prod_{y\in \overrightarrow{S(x)}}K_{<x,y>}\prod_{y\in \overrightarrow{S(x)}}h^{(y)}\prod_{y\in \overleftarrow{S(x)}}K^{*}_{<x,y>}\right]=h^{(x)}
\ \ \textrm{for every} \ \  x\in L.
\end{eqnarray}
Then the functionals $\{\ffi^{(n,f)}_{w_0,\bh}\}$ satisfy the
compatibility condition \eqref{compatibility}. Moreover, there is a
unique forward quantum $d$-Markov chain \footnote{For the definition
of quantum Markov chain we refer \cite{AMSa}.}
$\ffi^{(b)}_{w_0,{\bh}}$ on $\cb_L$ such that
$\ffi^{(f)}_{w_0,{\bh}}=w-\lim_{n\to\infty}\ffi^{(n,f)}_{w_0,\bh}$.
\end{theorem}

From direct calculation we can derive the following

\begin{proposition}\label{state^nwithW_n}
If \eqref{eq1} and \eqref{eq2} are satisfied then one has $\ffi^{(n,f)}_{w_0,\bh}(a)=\tr(\cw_{n]}(a))$ for any $a\in \cb_{\Lambda_n}$.
\end{proposition}

In \cite{AMSaII} we have introduced a notion of the phase transition
quantum Markov chains associated with the given family
$\{K_{<x,y>}\}$ of operators. Heuristically, such aphase transition
means the existence of two distinct QMC for the given
$\{K_{<x,y>}\}$. Let us provide a more exact definition.

\begin{definition}
We say that there exists a phase transition for a family of
operators $\{K_{<x,y>}\}$ if \eqref{eq1}, \eqref{eq2} have at least
two $(u_0,\{h_x\}_{x\in L})$ and $(v_0,\{s_x\}_{x\in L})$ solutions
such that the corresponding quantum $d$-Markov chains
$\ffi_{u_0,\bh}$ and $\ffi_{v_0,\bs}$ are not quasi equivalent.
Otherwise, we say there is no phase transition.
\end{definition}

In \cite{AMSaII} we have established the existence of the phase
transition for $XY$-model on the Caylay tree of order three.


\section{QMC associated with Ising model and main results}\label{exam1}

In this section, we consider the Ising model and formulate the main
results of the paper. In what follows, we consider a semi-infinite
Cayley tree $\G^k=(L,E)$ of order $k$. Our starting $C^{*}$-algebra
is the same $\cb_L$ but with $\cb_{x}=M_{2}(\bc)$ for $x\in L$. By
$\s_x^{(u)},\s_y^{(u)},\s_z^{(u)}$ we denote the Pauli spin
operators at site $u\in L$. Here are they
\begin{equation}\label{pauli}
\id^{(u)}=\left(
          \begin{array}{cc}
            1 & 0 \\
            0 & 1 \\
          \end{array}
        \right), \quad
\s_x^{(u)}= \left(
          \begin{array}{cc}
            0 & 1 \\
            1 & 0 \\
          \end{array}
        \right), \quad
\s_y^{(u)}= \left(
          \begin{array}{cc}
            0 & -i \\
            i & 0 \\
          \end{array}
        \right), \quad
\s_z^{(u)}= \left(
          \begin{array}{cc}
            1 & 0 \\
            0 & -1 \\
          \end{array}
        \right).
\end{equation}

For every edge $<u,v>\in E$ put
\begin{equation}\label{1Kxy1}
K_{<u,v>}=\exp\{\b H_{<u,v>}\}, \ \ \b>0,
\end{equation}
where
\begin{equation}\label{1Hxy1}
H_{<u,v>}=\frac{1}{2}\big(\id^{(u)}\id^{(v)}+\s_{z}^{(u)}\s_{z}^{(v)}\big).
\end{equation}

Such kind of Hamiltonian is called {\it Ising model} per edge
$<x,y>$.

Now taking into account the following equalities
\begin{eqnarray*}\label{1Hxym}
&&H_{<u,v>}^{m}=H_{<u,v>}=\frac{1}{2}\big(\id^{(u)}\id^{(v)}+\s_{z}^{(u)}\s_{z}^{(v)}\big),\
\end{eqnarray*}
one finds
\begin{eqnarray}\label{K<u,v>}
K_{<u,v>}=\id^{(u)}\id^{(v)}+(\exp\beta-1)H_{<u,v>}.
\end{eqnarray}
It follows from \eqref{1Hxy1} that
\begin{eqnarray}\label{K<u,v>K_0K_3}
K_{<u,v>}=K_0\id^{(u)}\id^{(v)}+K_3\s_{z}^{(u)}\s_{z}^{(v)}.
\end{eqnarray}
where, $K_0=\frac{\exp\beta+1}{2},$ $K_3=\frac{\exp\beta-1}{2}.$

The main result of the present paper concerns the existence of the
phase transition for the model \eqref{1Kxy1}. Namely, we have the following result.

\begin{theorem}\label{main} Let $\{K_{<x,y>}\}$ be given by \eqref{1Kxy1}
on the Cayley tree of order $k\geq 2$ and $\theta=\exp\{2\beta\},$
$\beta>0$. Then the following assertions hold:
\begin{enumerate}
\item[(i)] If $\theta\leq\frac{k+1}{k-1}$ then  there is a unique forward
quantum $d$-Markov chain.
\item[(ii)] If $\theta>\frac{k+1}{k-1}$ then  there is a phase transition for a given model,
i.e. there are two distinct forward quantum $d$-Markov chains.
\end{enumerate}
\end{theorem}

\begin{remark} We point out that in the classical case, for the Ising
model on the Calley tree, the phase transition occurs if
$\theta>\frac{k+1}{k-1}$ (see for example \cite{Geor,Roz2}). This
means that our definition of the phase transition is compatible with
well-known definitions.
\end{remark}

The rest of the paper will be devote to the proof the this theorem.
To do it, we shall use a dynamical system approach, which is
associated with the equations \eqref{eq1},\eqref{eq2}.

\section{A dynamical system related to boundary conditions}

In this section we shall reduce equations \eqref{eq2} to
some dynamical system. Our goal is to describe all solutions
$\bh=\{h_x\}$ of that equation.

Furthermore, we shall assume that $h_u=h_v$ for every $u,v\in W_n$,
$n\in\bn$. Hence, we denote $h_x^{(n)}:=h^{(x)}$ and $h_y^{(n+1)}:=h^{(y)}$ if $x\in W_n$ and $y\in W_{n+1}$. Now
from \eqref{1Kxy1},\eqref{1Hxy1} one can see that
$K_{<u,u>}=K^{*}_{<u,v>}$, therefore, equation \eqref{eq2} can be
rewritten as follows
\begin{eqnarray}\label{Isingstate}
\tr_{x]}\left[\prod_{y\in \overrightarrow{S(x)}}K_{<x,y>}\prod_{y\in \overrightarrow{S(x)}}h_y^{(n+1)}\prod_{y\in \overleftarrow{S(x)}}K_{<x,y>}\right]=h_x^{(n)}
\ \ \textrm{for every} \ \  x\in L.
\end{eqnarray}

Assume that
\begin{equation*}
h_{x}^{(n)}=\left(
          \begin{array}{cc}
            a^{(n)}_{11} & a^{(n)}_{12} \\
            a^{(n)}_{21} & a^{(n)}_{22} \\
          \end{array}
        \right), \quad\quad
h_{y}^{(n+1)}=\left(
          \begin{array}{cc}
            a^{(n+1)}_{11} & a^{(n+1)}_{12} \\
            a^{(n+1)}_{21} & a^{(n+1)}_{22} \\
          \end{array}
        \right).
\end{equation*}

Then after simple calculations, one can reduce \eqref{Isingstate} to
\begin{eqnarray*}
   \left(\frac{a^{(n+1)}_{11}\exp2\beta+a^{(n+1)}_{22}}{2}\right)^k &=& a^{(n)}_{11} \\
   0 &=& a^{(n)}_{12} \\
   0 &=& a^{(n)}_{21} \\
   \left(\frac{a^{(n+1)}_{11}+a^{(n+1)}_{22}\exp2\beta}{2}\right)^k &=& a^{(n)}_{22}
\end{eqnarray*}

 From the last system of equations we obtain
\begin{eqnarray}\label{a_12a_21}
a^{(n)}_{12}=a^{(n)}_{21}=0, \ \ \ \ \forall n\in \bn,
\end{eqnarray}
and
\begin{eqnarray*}
   a^{(n+1)}_{11} &=& \frac{2\exp2\beta\sqrt[k]{a^{(n)}_{11}}-2\sqrt[k]{a^{(n)}_{22}}}{\exp4\beta-1}\\
   a^{(n+1)}_{22} &=& \frac{-2\sqrt[k]{a^{(n)}_{11}}+2\exp2\beta\sqrt[k]{a^{(n)}_{22}}}{\exp4\beta-1}
\end{eqnarray*}

The last equalities allow us to consider a dynamical system
$f:\br^2_{+}\to \br^2_{+}$ given by

\begin{eqnarray}\label{dynamicalsysf}
\left\{ \begin{array}{l}
   x{'} = \dfrac{2\theta\sqrt[k]{x}-2\sqrt[k]{y}}{\theta^2-1}\\
   y{'} = \dfrac{-2\sqrt[k]{x}+2\theta\sqrt[k]{y}}{\theta^2-1}
\end{array}
\right.
\end{eqnarray}
where $\beta>0,$ $\theta=\exp\{2\beta\},$  and $(x{'},y{'})=f(x,y).$
Then one has that $\theta>1.$

\begin{remark}
The dynamical system $f:\br_{+}^2\to\br_{+}^2$ given by
\eqref{dynamicalsysf} is well-defined if and only if $x,y\ge 0$ and
$\frac{1}{\theta^k}y\le x \le \theta^k y .$ In what follows, we will
only consider  the case in which $x>0$ and $y>0$.
\end{remark}

\section{Asymptotical behavior of the dynamical system.}

In this section we shall find fixed points of \eqref{dynamicalsysf}
and prove the absence of periodic points. Moreover, we investigate
an asymptotical behavior it.

It is clear from \eqref{dynamicalsysf} that
$$\frac{x{'}}{y{'}}=\frac{2\theta\sqrt[k]{x}-2\sqrt[k]{y}}{-2\sqrt[k]{x}+2\theta\sqrt[k]{y}}
=\frac{\theta\sqrt[k]{\frac{x}{y}}-1}{\theta-\sqrt[k]{\frac{x}{y}}}.$$

Therefor, let us define the function $g_\theta:\br\to\br$ as follows
\begin{eqnarray}\label{gthetaoft}
g_\theta(t)=\frac{\theta\sqrt[k]{t}-1}{\theta-\sqrt[k]{t}}.
\end{eqnarray}

One can see that the domain of $g_\theta$ is
$\Delta:=[0,\theta^k)\cup(\theta^k,\infty).$

Let us study an asymptotical behavior of the function $g_\theta:\Delta\to\br.$

\begin{proposition}\label{monotonicityofg}
The function $g_\theta:\Delta\to\br$ given by \eqref{gthetaoft} is
strictly increasing.
\end{proposition}
\begin{proof}
Let us calculate the derivative of the function $g_\theta.$
\begin{eqnarray*}
g_\theta'(t)&=&\frac{\frac{\theta}{2\sqrt[k]{t}}(\theta-\sqrt[k]{t})+\frac{1}{2\sqrt[k]{t}}(\theta\sqrt[k]{t}-1)}{(\theta-\sqrt[k]{t})^2}
=\frac{\theta^2-1}{2\sqrt[k]{t}(\theta-\sqrt[k]{t})^2}.
\end{eqnarray*}
Since $\theta>1,$ hence $g_\theta'(t)>0$ for all $t\in \Delta.$ This
means that the function $g_\theta:\Delta\to\br$ given by
\eqref{gthetaoft} is increasing in its domain $\Delta$.
\end{proof}

\begin{remark} Let $g_\theta:\Delta\to\br$ be a given function by \eqref{gthetaoft}. Then one can easily check that:
\begin{itemize}
\item [$\rm(i)$]\label{periodicpointofg}
The function $g_\theta$ does not
have any $m$ periodic point in its domain $\Delta,$ where $m>1;$
\item [$\rm(ii)$] \label{positivedomainofg}
The function $g_\theta$ is positive if and only if
$t\in\left(\dfrac{1}{\theta^k},\theta^k\right).$
\end{itemize}
\end{remark}

\begin{proposition}\label{numberoffixedpointofg}
Let $g_\theta:D\to\br$ be a function given by \eqref{gthetaoft}.
Then the following assertions hold:
\begin{itemize}
  \item [$\rm(i)$] If $\theta>\frac{k+1}{k-1}$ then it has three fixed points which are $t_1,t_2, t_3$ such that $$\frac{1}{\theta^k}<t_2<t_1=1<t_3<\theta^k;$$
  \item [$\rm(ii)$] If $1<\theta\le \frac{k+1}{k-1}$ then it has a unique fixed point which is $t_1=1.$
\end{itemize}
\end{proposition}

\begin{proof} In order to find all fixed points of $g_\theta$
we should solve the following equation
\begin{eqnarray}\label{equationforfixedpoint}
\frac{\theta\sqrt[k]{t}-1}{\theta-\sqrt[k]{t}}= t.
\end{eqnarray}

Let us denote
$$x=\theta\frac{\theta\sqrt[k]{t}-1}{\theta-\sqrt[k]{t}}.$$ After
some algebraic manipulations the equation
\eqref{equationforfixedpoint} takes the following form
\begin{eqnarray}\label{equationfromPreston}
\frac{x}{\theta^k}=\left(\frac{x+1}{x+\theta^2}\right)^k.
\end{eqnarray}
In this case, we should find all positive solutions of
\eqref{equationfromPreston}. It was shown in \cite{[Pr]} that if
$\theta>\frac{k+1}{k-1}$ then \eqref{equationfromPreston} has three
positive solutions which are $x_1, x_2,x_3$ such that
$0<x_2<x_1=\theta<x_3$ and if $1<\theta\le\frac{k+1}{k-1}$ then the
equation \eqref{equationfromPreston} has a unique solution which is
$x_1=\theta.$ Then the corresponding solutions of
\eqref{equationforfixedpoint} are $t_1,t_2,t_3.$ This completes the
proof.
\end{proof}

\begin{proposition}\label{inequalityforgtheta}
Let $g_\theta:\Delta\to\br $ be a function given by
\eqref{gthetaoft}. Then the following assertions hold:
\begin{itemize}
  \item [$\rm(i)$] Let $\theta>\frac{k+1}{k-1}$. Then one has
\begin{enumerate}
\item[(a1)] $g_{\theta}(t)<t$ for any $t\in
(\frac{1}{\theta^k},t_2)\cup(t_1,t_3)$ and $g_{\theta}(t)>t$ for
any $t\in (t_2,t_1)\cup(t_3,\theta^k);$

\item [(b1)] If $t_0\in (t_2,t_3)$ then the trajectory
$\{g_\theta^{(n)}(t_0)\}_{n=1}^\infty,$ starting from the point
$t_0,$ converges to the fixed point $t_1$ which is equal to one;
  \item [(c1)] If $t_0\in (\frac{1}{\theta^k},t_2)\cup(t_3,\theta^k)$ then the
trajectory $\{g_\theta^{(n)}(t_0)\}_{n=1}^\infty,$ starting from the
point $t_0,$ is finite.
\end{enumerate}

\item [$\rm(ii)$] Let $1<\theta\le \frac{k+1}{k-1}$. Then one has
\begin{enumerate}
\item[(a2)]  $g_{\theta}(t)<t$ for any $t\in (\frac{1}{\theta^k},t_1)$ and
$g_{\theta}(t)>t$ for any $t\in (t_1,\theta^k).$

\item[(b2)] for any initial point $t_0\in
(\frac1{\theta^k},\theta^k),$ the trajectory
$\{g_\theta^{(n)}(t_0)\}_{n=1}^\infty$ starting from  the point
$t_0$ is finite.
\end{enumerate}
\end{itemize}
\end{proposition}

\begin{proof} (i) Assume that $\theta >\frac{k+1}{k-1}$. Let us prove (a1).  One
can see that
\begin{eqnarray}\label{comparegoftandt}
g_\theta(t)-t=\frac{(\sqrt[k]{t}-\sqrt[k]{t_1})(\sqrt[k]{t}-\sqrt[k]{t_2})(\sqrt[k]{t}-\sqrt[k]{t_3})}{\theta-\sqrt[k]{t}}\phi(\sqrt[k]{t})
\end{eqnarray}
where $\phi(\sqrt[k]{t})$ is a polynomial of the argument $\sqrt[k]{t}$ and $\phi(\sqrt[k]{t})>0$ for any $t\in(\frac{1}{\theta^k},\theta^k).$

It is clear from \eqref{comparegoftandt} that if $t\in
(\frac{1}{\theta^k},t_2)\cup(t_1,t_3)$ then $g_{\theta}(t)<t$ and
if $t\in (t_2,t_1)\cup(t_3,\theta^k)$ then $g_{\theta}(t)>t.$

(b1) Assume that $t_0\in (t_2,t_3).$ Since the function $g_\theta$
is strictly increasing and $t_1,t_2,t_3$ such that
$\frac{1}{\theta^k}<t_2<t_1<t_3<\theta^k$ are its fixed points, the
segments $(t_2,t_1),$ $(t_1,t_3),$ $(t_2,t_3)$ are invariant w.r.t.
the function $g_\theta$ and $g_{\theta}(t)>0$ for any $t\in
(t_2,t_3).$ Therefore, the trajectory
$\{g_\theta^{(n)}(t_0)\}_{n=1}^\infty,$ starting from the point
$t_0,$ is well-defined.

Without loss any generality, we assume that $t_0\in (t_2,t_1).$
According to (a1) we have $g_{\theta}(t_0)>t_0.$ Since the function
$g_\theta$ is strictly increasing, one finds
$$t_2<t_0<g_\theta(t_0)<g_\theta^{(2)}(t_0)<\cdots<g_\theta^{(n)}(t_0)<\cdots<t_1.$$
Then $\{g_\theta^{(n)}(t_0)\}_{n=1}^\infty$ is a convergent sequence
and its limiting point should be a fixed point which is equal to
$t_1=1$.

Analogously, one can show that if $t_0\in(t_1,t_3)$ then the
trajectory $\{g_\theta^{(n)}(t_0)\}_{n=1}^\infty,$ starting from the
point $t_0,$ is a monotone decreasing sequence on $(t_1,t_3)$ and
its limiting point is a fixed point $t_1=1$.

(c1) Now assume that $t_0\in
(\frac{1}{\theta^k},t_2)\cup(t_3,\theta^k).$

Without loss of any generality we suppose that $t_0\in
(\frac{1}{\theta^k},t_2).$ Let us assume that the trajectory
$\{g_\theta^{(n)}(t_0)\}_{n=1}^\infty$  is well-defined and it has
an infinite number of distinct terms. According to Remark
\ref{positivedomainofg} (ii) all terms of the trajectory should be
inside of the interval $(\frac{1}{\theta^k},\theta^k),$ i.e.,
\begin{eqnarray}\label{ggreatfrac1overtheta}
g_\theta^{(n)}(t_0)>\frac1{\theta^k},
\end{eqnarray} for any $n\in\bn.$

On the other hand, since $g_\theta:\Delta\to\br$ is increasing, due
to (a1) one gets
\begin{eqnarray}\label{decriassequenceofg}
t_2>t_0>g_\theta(t_0)>g_\theta^{(2)}(t_0)>\cdots>g_\theta^{(n)}(t_0)>\cdots>\frac{1}{\theta^k}.
\end{eqnarray}
It follows from \eqref{decriassequenceofg} that the sequence
$\{g_\theta^{(n)}(t_0)\}_{n=1}^\infty$ converges and its limiting
point should be a fixed point which is less than $t_2.$ However, the
function $g_\theta$ does not have any fixed points except
$t_1,t_2,t_3.$ This contradiction shows that the trajectory
$\{g_\theta^{(n)}(t_0)\}_{n=1}^\infty,$ starting from the point
$t_0,$ is finite.

(ii) Let $1<\theta\le \frac{k+1}{k-1}$. Then (a2) is evident. Using
the same argument as (b1) one can prove (b2).

This completes the proof.
\end{proof}

Let us study an asymptotical behavior of the dynamical system
$f:\br_{+}^2\to\br_{+}^2$ given by \eqref{dynamicalsysf}

\begin{proposition}\label{numberoffixedpointofdynsysoff}
Let  $f:\br_{+}^2\to\br_{+}^2$ be a dynamical system given by
\eqref{dynamicalsysf}. Then the following assertions hold:
\begin{itemize}
  \item [$\rm(i)$] If $\theta>\frac{k+1}{k-1}$ then it has three fixed points
which are equal  to $\left(A_i\sqrt[k-1]{B_i}; B_i\sqrt[k-1]{B_i}\right),$ where
$$A_i=\frac{2\theta\sqrt[k]{t_i}-2}{\theta^2-1}, \ \ \ B_i=\frac{2\theta-2\sqrt[k]{t_i}}{\theta^2-1}, \ \ i=1,2,3,$$
  \item [$\rm(ii)$] If $1<\theta\le \frac{k+1}{k-1}$ then it
has one fixed point which is equal to $(A_1\sqrt[k-1]{B_1},B_1\sqrt[k-1]{B_1})$.
\end{itemize}
\end{proposition}

\begin{proof} From \eqref{dynamicalsysf} we immediately find
\begin{eqnarray}\label{ratioofx'andy'}
\frac{x{'}}{y{'}}&=&\frac{\theta\sqrt[k]{x}-\sqrt[k]{y}}{-\sqrt[k]{x}+\theta\sqrt[k]{y}}=
\frac{\theta\sqrt[k]{\frac{x}{y}}-1}{-\sqrt[k]{\frac{x}{y}}+\theta}.
\end{eqnarray}
Hence, fixed points of $f:\br_{+}^2\to\br_{+}^2$ should satisfy the
following equation
\begin{eqnarray}\label{ratioofxandy}
\frac{x}{y}&=&\frac{\theta\sqrt[k]{\frac{x}{y}}-1}{\theta-\sqrt[k]{\frac{x}{y}}}
\end{eqnarray}

Denote $t=\frac{x}{y}$. Then \eqref{ratioofxandy} takes the form
$g_\theta(t)=t.$ Consequently, Proposition
\ref{numberoffixedpointofg} implies that if $\theta>\frac{k+1}{k-1}$
then the function $g_\theta$ has three fixed points which are equal
to $t_1,t_2,t_3$ and if $1<\theta\le \frac{k+1}{k-1}$ then the
function $g_\theta$ has one fixed point which is equal to $t_1$.
Therefore, we have that if $\theta>\frac{k+1}{k-1}$ then
$$\frac{x}{y}=t_i, \ \ \ i=1,2,3$$
and if $1<\theta\le \frac{k+1}{k-1}$
then $$\frac{x}{y}=t_1.$$
After elemental calculation one can see that if $\theta>\frac{k+1}{k-1}$ then the
dynamical system $f:\br_{+}^2\to\br_{+}^2$ has three fixed points
which are equal  to $(A_i\sqrt[k-1]{B_i}; B_i\sqrt[k-1]{B_i}),$ where
$$A_i=\frac{2\theta\sqrt[k]{t_i}-2}{\theta^2-1}, \ \ \ B_i=\frac{2\theta-2\sqrt[k]{t_i}}{\theta^2-1}, \ \ i=1,2,3,$$
and if $1<\theta\le \frac{k+1}{k-1}$ then it has one fixed point which is equal to $(A_1\sqrt[k-1]{B_1}; B_1\sqrt[k-1]{B_1})$.
\end{proof}

\begin{proposition}
Let $f:\br_{+}^2\to\br_{+}^2$ be  a dynamical system given by
\eqref{dynamicalsysf}. Then the following assertions hold true:
\begin{itemize}
  \item [$\rm(i)$] If $\theta>\frac{k+1}{k-1}$ then it has three invariant
semi-lines $l_i$ which are defined by $y=\frac{1}{t_{i}}x,$ $i=1,2,3$
  \item [$\rm(ii)$] If $1<\theta\le \frac{k+1}{k-1}$ then
it has one invariant semi-line $l_1$ which is defined by
$y=\frac{1}{t_{1}}x.$
\end{itemize}
\end{proposition}

\begin{proof} It follows from \eqref{ratioofx'andy'} that if $t^{*}$ is a fixed
point of the function $g_\theta,$ then $\dfrac{x}{y}=t^{*}$ yields
$\dfrac{x{'}}{y{'}}=t^{*}.$ Therefore,  if $\theta>\frac{k+1}{k-1}$
then the dynamical system $f:\br_{+}^2\to\br_{+}^2$ has three
invariant semi-lines given by $y=\frac{1}{t_{i}}x,$ $i=1,2,3$ and if
$1<\theta\le \frac{k+1}{k-1}$ then it has one invariant semi-line
defined by $y=\frac{1}{t_{1}}x.$
\end{proof}

\begin{theorem}\label{incasethetagratk1k1}
Let $f:\br_{+}^2\to\br_{+}^2$ a dynamical system  be given by
\eqref{dynamicalsysf} and $\theta>\frac{k+1}{k-1}.$ Then the following assertions hold true:
\begin{itemize}
  \item[\rm(i)] If an initial point $(x^0,y^0)$ belongs to an invariant semi-line $l_i$
(where $i=1,2,3$) of the dynamical system $f:\br_{+}^2\to\br_{+}^2$
then the trajectory $\{f^{(n)}(x^0,y^0)\}_{n=1}^\infty,$ starting
from the point $(x^0,y^0),$ converges to a fixed point
$(A_i\sqrt[k-1]{B_i}; B_i\sqrt[k-1]{B_i})$ which belongs to an invariant line $l_i;$
  \item[\rm(ii)] If an initial point $(x^0,y^0)$ satisfies the following condition
$$\frac{x^0}{y^0}\in(t_2,t_1)\cup(t_1,t_3)$$ then the trajectory
$\{f^{(n)}(x^0,y^0)\}_{n=1}^\infty,$ starting from the point
$(x^0,y^0),$ converges to a fixed point $(A_1\sqrt[k-1]{B_1}; B_1\sqrt[k-1]{B_1})$ which
belongs to an invariant semi-line $l_1.$
  \item[\rm(iii)]If an initial point $(x^0,y^0)$ satisfies the following condition
$$\frac{x^0}{y^0}\in\left(\frac{1}{\theta^k},t_2\right)\cup(t_3,\theta^k)$$ then the trajectory
$\{f^{(n)}(x^0,y^0)\}_{n=1}^\infty,$ starting from the point
$(x^0,y^0),$ is finite.
\end{itemize}
\end{theorem}

\begin{proof} Assume that $\theta>\frac{k+1}{k-1}.$

(i). Now we suppose that $(x^0,y^0)$ belongs to an invariant
semi-line $l_i$ ($i=1,2,3$) of the dynamical system
$f:\br_{+}^2\to\br_{+}^2.$ Then one has

$$\frac{x^{(n)}}{y^{(n)}}=t_i,$$
for any $n\in\bn,$ where $(x^{(n)},y^{(n)})\equiv f^{(n)}(x^0,y^0).$
Hence, one finds
$$
\left\{
\begin{array}{l}
x^{(n)}=A_i\sqrt[k]{y^{(n-1)}}\\
y^{(n)}=B_i\sqrt[k]{y^{(n-1)}}
\end{array}
\right.
$$
Therefore
$$
\left\{
\begin{array}{l}
x^{(n)}=\underbrace{A_i\sqrt[k]{B_i\sqrt[k]{B_i\cdots\sqrt[k]{B_i}}}}_{n}\sqrt[\leftroot{-5}\uproot{5}k^{n}]{y^0}
         =A_iB_i^{\dfrac{\frac1k-\frac{1}{k^{n}}}{1-\frac1k}}\sqrt[\leftroot{-5}\uproot{5}k^{n}]{y^0}\\
y^{(n)}=\underbrace{B_i\sqrt[k]{B_i\sqrt[k]{B_i\cdots\sqrt[k]{B_i}}}}_{n}\sqrt[\leftroot{-5}\uproot{5}k^{n}]{y^0}
         =B_i^{\dfrac{1-\frac{1}{k^{n}}}{1-\frac1k}}\sqrt[\leftroot{-5}\uproot{5}k^{n}]{y^0}.
\end{array}
\right.
$$
If we take into account $\dfrac{x_0}{A_i\sqrt[k-1]{B_i}}=\dfrac{y_0}{B_i\sqrt[k-1]{B_i}}$ then
\begin{eqnarray}\label{tracjectoryofxnyn}
\left\{
\begin{array}{l}
x^{(n)}=A_i\sqrt[k-1]{B_i}\sqrt[\leftroot{-5}\uproot{5}k^{n}]{\dfrac{x^0}{A_i\sqrt[k-1]{B_i}}}\\
y^{(n)}=B_i\sqrt[k-1]{B_i}\sqrt[\leftroot{-5}\uproot{5}k^{n}]{\dfrac{y^0}{B_i\sqrt[k-1]{B_i}}}
\end{array}
\right.
\end{eqnarray}
It is clear that the sequence $(x^{(n)},y^{(n)})$ converges to
$(A_i\sqrt[k-1]{B_i};B_i\sqrt[k-1]{B_i})$ which is a fixed point
belonging to $l_i.$

$(ii).$ Assume that an initial point $(x^0,y^0)$ satisfies
$$\frac{x^0}{y^0}\in(t_2,t_1)\cup(t_1,t_3).$$
It follows from \eqref{ratioofx'andy'} that
$$\frac{x^{(n)}}{y^{(n)}}=g\left(\frac{x^{(n-1)}}{y^{(n-1)}}\right).$$
According to Proposition \ref{inequalityforgtheta} the sequence
$\left\{\dfrac{x^{(n)}}{y^{(n)}}\right\}_{n=1}^\infty$ converges to
the fixed point $t_1$ of the function $g_\theta:\Delta\to\br.$

Taking $\dfrac{x^{(n)}}{y^{(n)}}=c_n,$ then one gets that
$g_\theta(c_n)=c_{n+1}$ and
$$
\left\{
\begin{array}{l}
x^{(n+1)}=a_{n}\sqrt[k]{y^{(n)}}\\
y^{(n+1)}=b_{n}\sqrt[k]{y^{(n)}}
\end{array}
\right.
$$
where $$a_n=\frac{2\theta\sqrt[k]{c_n}-2}{\theta^2-1},\ \ \
b_n=\frac{2\theta-2\sqrt[k]{c_n}}{\theta^2-1}.$$ So, we find that
$$
\left\{
\begin{array}{l}
x^{(n+1)}=\underbrace{a_n\sqrt[k]{b_{n-1}\sqrt[k]{b_{n-2}\cdots\sqrt[k]{b_{0}}}}}_{n+1}\sqrt[\leftroot{-5}\uproot{5}k^{n+1}]{y^0}\\
y^{(n+1)}=\underbrace{b_n\sqrt[k]{b_{n-1}\sqrt[k]{b_{n-2}\cdots\sqrt[k]{b_0}}}}_{n+1}\sqrt[\leftroot{-5}\uproot{5}k^{n+1}]{y^0}
\end{array}
\right.
$$
The following lemma is useful to calculate the limiting point of the
sequence $\{(x^{(n)},y^{(n)})\}_{n=0}^\infty.$
\begin{lemma}\label{sequenceofdn}
If a sequence $\{b_n\}_{n=0}^\infty,$ with positive terms,
converges to $\beta_0>0$ then the sequence
$$\beta_n=\underbrace{b_n\sqrt[k]{b_{n-1}\sqrt[k]{b_{n-2}\cdots\sqrt[k]{b_0}}}}_{n+1}$$
converges to $\beta_0\sqrt[k-1]{\beta_0}.$
\end{lemma}

We know that $$c_{n}\to t_1, \ \ \ b_n\to B_1,\ \ \ a_n\to A_1.$$
Then, according to Lemma \ref{sequenceofdn}, the sequence
$(x^{(n)},y^{(n)})$ converges to
$\left(A_1\sqrt[k-1]{B_1},B_1\sqrt[k-1]{B_1}\right)$ which belongs
to $l_1.$

$(iii).$ Now assume that an initial point $(x^0,y^0)$ satisfies
$$\frac{x^0}{y^0}\in\left(\frac{1}{\theta^k},t_2\right)\cup(t_3,\theta^k).$$

It follows from \eqref{ratioofx'andy'} that
$$\frac{x^{(n+1)}}{y^{(n+1)}}=g\left(\frac{x^{(n)}}{y^{(n)}}\right),$$ for any $n\in\bn.$
According to Proposition \ref{inequalityforgtheta} the sequence
$\left\{\dfrac{x^{(n)}}{y^{(n)}}\right\}_{n=1}^\infty$ has a finite
number of terms. Therefore the sequence
$\left\{(x^{(n+1)},y^{(n+1)})\right\}_{n=0}^\infty$ is finite.
\end{proof}

Analogously, one can prove the following

\begin{theorem}\label{incasethetalessk1k1}
Let $f:\br_{+}^2\to\br_{+}^2$ be a dynamical system given by
\eqref{dynamicalsysf} and $1<\theta\le \frac{k+1}{k-1}.$
\begin{itemize}
  \item[(i)] If an initial point $(x^0,y^0)$ belongs to an invariant semi-line $l_1$
of the dynamical system $f:\br_{+}^2\to\br_{+}^2$ then the
trajectory $\{f^{(n)}(x^0,y^0)\}_{n=1}^\infty,$ starting from the
point $(x^0,y^0),$ converges to a fixed point $(A_1\sqrt[k-1]{B_1}; B_1\sqrt[k-1]{B_1})$ which
belongs to an invariant semi-line $l_1.$
  \item[(ii)] If an initial point $(x^0,y^0)$ satisfies the following condition
$$\frac{x^0}{y^0}\in\left(\frac{1}{\theta^k},t_1\right)\cup(t_1,\theta^k)$$ then the trajectory
$\{f^{(n)}(x^0,y^0)\}_{n=1}^\infty,$ starting from the point
$(x^0,y^0),$ is finite.
\end{itemize}
\end{theorem}

\begin{remark}\label{aninvariantline}
One can easily check that if an initial point $(x^{0},y^{0})$ belongs to an invariant line $l_i$ of the dynamical system \eqref{dynamicalsysf} then the trajectory $(x^{(n)},y^{(n)})$ starting from $(x^{0},y^{0})$ has the following form
\begin{eqnarray}\label{trajectoryanytheta}
\left\{
\begin{array}{l}
x^{(n)}=A_i\sqrt[k-1]{B_i}\sqrt[\leftroot{-5}\uproot{5}k^{n}]{\dfrac{x^0}{A_i\sqrt[k-1]{B_i}}}\\
y^{(n)}=B_i\sqrt[k-1]{B_i}\sqrt[\leftroot{-5}\uproot{5}k^{n}]{\dfrac{y^0}{B_i\sqrt[k-1]{B_i}}}
\end{array}
\right.
\end{eqnarray}
In the case $\theta>\frac{k+1}{k-1}$ the formula
\eqref{trajectoryanytheta} was already shown in
\eqref{tracjectoryofxnyn} where $i=1,2,3$. In the case
$\theta\leq\frac{k+1}{k-1}$ the dynamical system
\eqref{dynamicalsysf} has unique an invariant line $l_1$ at $i=1$ in
the formula \eqref{trajectoryanytheta}.
\end{remark}

\section{Diagonalizability of forward QMC}

In previous section we have found fixed points of the dynamical
system \eqref{dynamicalsysf} and prove the absence of periodic
points for any $\theta>1$. Moreover, we investigated an asymptotical
behavior of \eqref{dynamicalsysf}. It is clear that every fixed
point of \eqref{dynamicalsysf} defines boundary conditions which are
solutions of \eqref{eq1}, \eqref{eq2}. Namely, we have that if
$\theta\leq\frac{k+1}{k-1}$ then  there are boundary conditions
$(w_0(\a_0),\{h_x(\a_0)\})$ of the model \eqref{K<u,v>}
\begin{eqnarray}\label{boundaryconditionsthetaalpha}
w_0(\a_0)=\frac{1}{\a_0}\sigma_0, \quad h_x(\a_0)=\a_0\sigma^{(x)}_0
\end{eqnarray}
and if $\theta>\frac{k+1}{k-1}$ then apart the previous one, there
are two extra boundary conditions $(w_0(\beta),\{h_x(\beta)\})$ and
$(w_0(\gamma),\{h_x(\gamma)\})$ of the model \eqref{K<u,v>}
\begin{eqnarray}
\label{boundaryconditionsthetabeta}&&w_0(\beta)=\frac{1}{\beta_0}\sigma_0, \quad h_x(\beta)=\beta_0\sigma^{(x)}_0+\beta_3\sigma_3^{(x)}\\
\label{boundaryconditionsthetagamma}&&w_0(\gamma)=\frac{1}{\gamma_0}\sigma_0, \quad h_x(\gamma)=\gamma_0\sigma^{(x)}_0+\gamma_3\sigma_3^{(x)}
\end{eqnarray}
here $\a_0=\sqrt[k-1]{\theta^k}$, and  $\beta=(\beta_0,\beta_3),$
$\gamma=(\gamma_0,\gamma_3)$ such that
\begin{eqnarray*}
&&\beta_0=\frac{A_2\sqrt[k-1]{B_2}+B_2\sqrt[k-1]{B_2}}{2}, \quad \beta_3=\frac{A_2\sqrt[k-1]{B_2}-B_2\sqrt[k-1]{B_2}}{2},\\
&&\gamma_0=\frac{A_3\sqrt[k-1]{B_3}+B_3\sqrt[k-1]{B_3}}{2}, \quad \gamma_3=\frac{A_3\sqrt[k-1]{B_3}-B_3\sqrt[k-1]{B_3}}{2}.
\end{eqnarray*}

Note that these boundary conditions
\eqref{boundaryconditionsthetaalpha},
\eqref{boundaryconditionsthetabeta},
\eqref{boundaryconditionsthetagamma} due to Theorem \ref{compa}
define the forward QMC. Hence, the existence of the boundary
conditions imply the existence of forward QMC for the model
\eqref{K<u,v>} for any $\theta>1$.

We are going to prove diagonalizability of the forward QMC corresponding
to any boundary conditions which are solutions of \eqref{eq1} and \eqref{eq2}.

Recall that the diagonal subalgebra $M_2^{d}(\bc)$ of the algebra
$M_2(\bc)$ is defined as follows
\begin{eqnarray*}
M_2^{d}(\bc)=\left\{a\in M_2(\bc): a=a_0\id+a_3\sigma_z\right\}.
\end{eqnarray*}
Since the elements $\id,\sigma_x,\sigma_y,\sigma_z$ are basis in
$M_2(\bc)$ then every element $a\in M_2(\bc)$ can be written in the
following form $a=a_0\id+a_1\sigma_x+a_2\sigma_y+a_3\sigma_z.$ Then
for any $a\in M_2(\bc)$ an element $a_d=a_0\id+a_3\sigma_z$ is
called \emph{its diagonal part} and an element
$a_{xy}=a_1\sigma_x+a_2\sigma_y$ is called \emph{its $xy-$part}. It
is clear that a linear span of the elements $\id,\sigma_z$ is a
commutative diagonal subalgebra $M_2^{d}(\bc).$ In these notions,
any element $a\in M_2(\bc)$ can be written as $a=a_d+a_{xy}.$

Let us consider a conditional expectation $E:M_2(\bc)\to
M_2^{d}(\bc)$ defined by
$$E(a)=e_{11}ae_{11}+e_{22}ae_{22},$$
where $e_{11},e_{22}$ are two minimal projectors of the algebra $M_2(\bc).$ It is clear that $E(a)=a^d.$

The diagonal subalgebra $\cb_L^d$ of the full
algebra $\cb_L$ is defined by an inductive limit
\begin{eqnarray*}
\cb_L^d=\overline{\bigcup\limits_{\L_n}\cb_{\L_n}^d}^{\|\cdot\|},\quad \cb_{\L_{n}}^d=\bigotimes\limits_{u\in\L_{n}}\cb_u^d, \quad
\cb_{u}^d=M_{2}^d(\bc),\ \forall \ u\in\L.
\end{eqnarray*}

We will define a conditional expectation $\ce:\cb_L\to \cb_L^d$ on
the algebra $\cb_L$ as follows: the value of this expectation at any
linear generator $a_{\L_n}=\bigotimes\limits_{x\in\L_n}a_x$ of the
algebra $\cb_{\L_n}$  is defined by
\begin{eqnarray}\label{expectgenerator}
\ce(a_{\L_n})=\bigotimes\limits_{x\in\L_n}E(a_x)
\end{eqnarray}
 and  the linear extension of \eqref{expectgenerator} defines
 it on the whole algebra $\cb_{\L_n}.$  By inductive limit we define the expectation $\ce$ on the full algebra $\cb_L.$ For any $a_{\L}\in\cb_{\L}$ the value $\ce(a_{\L})$ is called \emph{a diagonal part} of $a_{\L}$ and denote by $a^d_{\L}=\ce(a_{\L}).$

\begin{theorem}\label{diagonaliziablity}
Let $\ffi^{(f)}_{w_0,{\bh}}$ be a forward QMC of the model
\eqref{K<u,v>} with boundary conditions which are solutions of
\eqref{eq1} and \eqref{eq2}. Let $\ce:\cb_L\to \cb_L^d$ be a
conditional expectation given \eqref{expectgenerator} and
$\theta>1.$ Then for any $a\in\cb_{L}$ one has
\begin{eqnarray*}
\ffi^{(f)}_{w_0,{\bh}}(a)=\ffi^{(f)}_{w_0,{\bh}}(\ce(a)).
\end{eqnarray*}
\end{theorem}
\begin{proof}
Since $\cb_\L$ is a quasi-local algebra it is enough to show that
\begin{eqnarray}\label{phiE(a_n)}
\ffi^{(f)}_{w_0,{\bh}}(a_{\L_n})=\ffi^{(f)}_{w_0,{\bh}}(\ce(a_{\L_n})),
\end{eqnarray}
for any $a_{\L_n}\in\cb_{\L_n}.$  It follows from the definition of the QMC that
\begin{eqnarray*}
\ffi^{(f)}_{w_0,{\bh}}(a_{\L_n})=
w-\lim_{m\to\infty}\ffi^{(m,f)}_{w_0,\bh}(a_{\L_n})=\ffi^{(n,f)}_{w_0,\bh}(a_{\L_n}).
\end{eqnarray*}
Analogously, one has
\begin{eqnarray*}
\ffi^{(f)}_{w_0,{\bh}}(\ce(a_{\L_n}))=\ffi^{(n,f)}_{w_0,\bh}(\ce(a_{\L_n})).
\end{eqnarray*}
It follows from \eqref{a_12a_21} that any solutions $\{\bh_x\}_{x\in\L}$ of the equation \eqref{eq2} lie in the diagonal algebra $M^d_2(\bc)$
Let us choose $w_0=\frac{1}{Tr(h_0)}\id.$ Since an element $K_{<u,v>}$ given by \eqref{K<u,v>K_0K_3} lies in the diagonal algebra $M^d_2(\bc)\otimes M^d_2(\bc)$ then an element $\cw_n$ given by \eqref{w_n} lies in the diagonal algebra $\cb^d_{\L_n}.$ We then get
\begin{eqnarray*}
\ffi^{(n,f)}_{w_0,\bh}(\ce(a_{\L_n}))=\Tr(\cw_n(\ce(a_{\L_n}))=\Tr(\ce(\cw_na_{\L_n}))=\Tr(\cw_na_{\L_n})=\ffi^{(n,f)}_{w_0,\bh}(a_{\L_n}).
\end{eqnarray*}
This completes the proof.
\end{proof}

\section{Uniqueness of forward QMC: regime $\theta\le\frac{k+1}{k-1}$}

In this section we prove the first part of the main theorem (see
Theorem \ref{main}), i.e. we show the uniqueness of the forward
quantum $d$-Markov chain in the regime $1<\theta\le\frac{k+1}{k-1}$.

We assume that $\theta\le\frac{k+1}{k-1}.$ Since $t_1=1$ we have $\Theta:=A_1=B_1=\dfrac{2}{\theta+1}.$ Then it follows from Theorem \ref{incasethetalessk1k1} and Remark \ref{aninvariantline} that the equation \eqref{eq2} does not have any solution except the following parametrical solutions $\{h_x(\a)\}$ given by
\begin{equation}\label{solutionofmainstate}
h^{(n)}_x(\alpha)=\left(
                    \begin{array}{cc}
                    \sqrt[k-1]{\Theta^k}\sqrt[\leftroot{-5}\uproot{5}k^{n}]{\dfrac{\alpha}{\sqrt[k-1]{\Theta^k}}} & 0 \\
                      0 & \sqrt[k-1]{\Theta^k}\sqrt[\leftroot{-5}\uproot{5}k^{n}]{\dfrac{\alpha}{\sqrt[k-1]{\Theta^k}}} \\
                    \end{array}
                 \right),
                 \end{equation}
for every $x\in W_n$, here $\alpha$ is any positive real number and $n\in\bn\cup\{0\}$. One of the solutions of the equation \eqref{eq1} has the following form
\begin{equation}\label{foromegaalpha}
w_0(\alpha)=\left(
              \begin{array}{cc}
                \dfrac{1}{\alpha} & 0 \\
                0 & \dfrac{1}{\alpha} \\
              \end{array}
            \right)
\end{equation}

The boundary conditions corresponding to the fixed point $(\sqrt[k-1]{\Theta^k},\sqrt[k-1]{\Theta^k})$ of the dynamical system
\eqref{dynamicalsysf} are the following
\begin{equation}\label{solutionofmainstatewhenalphafixed}
w_0(\alpha_0)=\left(\begin{array}{cc}
                      \dfrac{1}{\sqrt[k-1]{\Theta^k}} & 0 \\
                      0 & \dfrac{1}{\sqrt[k-1]{\Theta^k}} \\
                    \end{array}
                  \right), \quad
                  h^{(n)}_x(\alpha_0)=\left(
                   \begin{array}{cc}
                \sqrt[k-1]{\Theta^k} & 0 \\
                0 & \sqrt[k-1]{\Theta^k} \\
              \end{array}
            \right), \ \ \forall x\in W_n,
            \end{equation}
which correspond to the value of $\alpha_0=\sqrt[k-1]{\Theta^k}$
in \eqref{solutionofmainstate}, \eqref{foromegaalpha}.

Let us consider the states $\ffi^{(n,f)}_{w_0(\a),\bh(\alpha)}$
corresponding to the solutions
$(w_0(\alpha),\{h_x^{(n)}(\alpha)\})$. By definition we have
\begin{eqnarray}\label{uniq}
\ffi^{(n,f)}_{w_0(\a),\bh(\alpha)}(x) &=&
\tr\left(w^{1/2}_{0}(\alpha)\prod_{i=0}^{n-1}K_{[i,i+1]}\prod_{x\in  \overrightarrow{W}_n}h^{(n)}_x(\alpha)
\prod_{i=0}^{n-1}K^{*}_{[n-1-i,n-i]}w^{1/2}_{0}(\alpha)x\right)\nonumber\\
&=&\frac{\left(\sqrt[k-1]{\Theta^k}\right)^{k^{n}}}{\alpha}\left(\sqrt[k^{n}]{\frac{\alpha}{\sqrt[k-1]{\Theta^k}}}\right)^{k^{n}}
\tr\left(\prod_{i=0}^{n-1}K_{[i,i+1]}\prod_{i=0}^{n-1}K^{*}_{[n-1-i,n-i]}x\right)\nonumber\\
&=&\frac{\alpha_0^{k^{n}}}{\alpha_0}\tr\left(\prod_{i=0}^{n-1}K_{[i,i+1]}\prod_{i=0}^{n-1}K^{*}_{[n-1-i,n-i]}x\right)\nonumber\\
&=& \tr\left(w^{1/2}_{0}(\alpha_0)\prod_{i=0}^{n-1}K_{[i,i+1]}\prod_{x\in  \overrightarrow{W}_n}h^{(n)}_x(\alpha_0)
\prod_{i=0}^{n-1}K^{*}_{[n-1-i,n-i]}w^{1/2}_{0}(\alpha_0)x\right)\nonumber\\
&=&\ffi^{(n,f)}_{w_0(\a_0),\bh(\alpha_0)}(x),
\end{eqnarray}
for any $\alpha$. Hence, from the definition of forward QMC one
finds that
$\ffi^{(f)}_{w_0(\a),\bh(\alpha)}=\ffi^{(f)}_{w_0(\a_0),\bh(\alpha_0)}$,
which yields that the uniqueness of QMC associated with the model
\eqref{1Kxy1}. Hence, Theorem \ref{main} (i) is proved.

\section{Existence of phase transition: regime $\theta>\frac{k+1}{k-1}$}

This section is devoted to the proof of part (ii) of Theorem
\ref{main}. In the sequel we suppose that $\theta>\frac{k+1}{k-1}.$

In this section, for the sake of simplicity of formulas, we will use
the following notations, for the Pauli matrices:
\begin{eqnarray*}
\sigma_0:=\id, \quad \sigma_1:=\sigma_x, \quad \sigma_2:=\sigma_y, \quad \sigma_3:=\sigma_z
\end{eqnarray*}

According to Proposition \ref{numberoffixedpointofdynsysoff} there
are three fixed points of the dynamical system \eqref{dynamicalsysf}
in the considering regime. Then the corresponding solutions of
equations \eqref{eq1},\eqref{eq2} can be written as follows:
$(w_0(\a_0),\{h_x(\a_0)\}),$ $(w_0(\beta),\{h_x(\beta)\})$ and
$(w_0(\gamma),\{h_x(\gamma)\})$, where
\begin{eqnarray*}
&&w_0(\a_0)=\frac{1}{\a_0}\sigma_0, \quad h_x(\a_0)=\a_0\sigma^{(x)}_0\\
&&w_0(\beta)=\frac{1}{\beta_0}\sigma_0, \quad h_x(\beta)=\beta_0\sigma^{(x)}_0+\beta_3\sigma_3^{(x)}\\
&&w_0(\gamma)=\frac{1}{\gamma_0}\sigma_0, \quad h_x(\gamma)=\gamma_0\sigma^{(x)}_0+\gamma_3\sigma_3^{(x)}
\end{eqnarray*}
here $\a_0=\sqrt[k-1]{\Theta^k}$, and  $\beta=(\beta_0,\beta_3),$ $\gamma=(\gamma_0,\gamma_3)$ are vectors with
\begin{eqnarray}
\label{beta03}&&\beta_0=\frac{A_2\sqrt[k-1]{B_2}+B_2\sqrt[k-1]{B_2}}{2}, \quad \beta_3=\frac{A_2\sqrt[k-1]{B_2}-B_2\sqrt[k-1]{B_2}}{2},\\
\label{gamma03}&&\gamma_0=\frac{A_3\sqrt[k-1]{B_3}+B_3\sqrt[k-1]{B_3}}{2}, \quad \gamma_3=\frac{A_3\sqrt[k-1]{B_3}-B_3\sqrt[k-1]{B_3}}{2}.
\end{eqnarray}

By $\ffi^{(f)}_{w_0(\a_0),{\bh(\a_0)}}$, $\ffi^{(f)}_{w_0(\beta),{\bh(\beta)}},$
and $\ffi^{(f)}_{w_0(\gamma),{\bh(\gamma)}}$ we denote the corresponding forward
quantum Markov chains. To prove the existence of the phase transition,
we need to show that there are two states which are not quasi-equivalent.
We will show that two states
$\ffi^{(f)}_{w_0(\a_0),{\bh(\a_0)}}$, $\ffi^{(f)}_{w_0(\gamma),{\bh(\gamma)}}$ are not
quasi-equivalent.  To do so, we will need some auxiliary facts and results.\\

First of all, we recall some properties of $2\times 2$  special matrices which are not required their proof.

Let $M,N$ be  matrices given as follows
\begin{eqnarray*}
M=\left(
  \begin{array}{cc}
    a & b \\
    b & a \\
  \end{array}
\right),\quad
N=\left(
  \begin{array}{cc}
    c & d \\
    d & c \\
  \end{array}
\right).
\end{eqnarray*}
Then, these matrices commute each other, i.e.,  $MN=NM$. For any
$n\in \bn$ one has
\begin{eqnarray}\label{M^n}
M^{n}=\dfrac{1}{2}\left(
                       \begin{array}{cc}
                         (a+b)^n+(a-b)^n & (a+b)^n-(a-b)^n \\
                         (a+b)^n-(a-b)^n & (a+b)^n+(a-b)^n \\
                       \end{array}
                     \right)
\end{eqnarray}

For the sake of simplicity we will use the following notations
\begin{eqnarray}\label{Theta+Theta-}
\Theta_{+}=\gamma_0\dfrac{\theta+1}{2}+\gamma_3\dfrac{\theta-1}{2}, \quad \Theta_{-}=\gamma_0\dfrac{\theta+1}{2}-\gamma_3\dfrac{\theta-1}{2},
\end{eqnarray}
where $\gamma_0,\gamma_3$ are given by \eqref{gamma03}. Let us denote by
\begin{equation}\label{mainformofmainMatrix}
\mathbb{A}=\dfrac{1}{2}\left(
    \begin{array}{cc}
       \frac{\theta+1}{2}\left(\Theta_{+}^{k-1}+\Theta_{-}^{k-1}\right) & \frac{\theta-1}{2}\left(\Theta_{+}^{k-1}-\Theta_{-}^{k-1}\right) \\
       \frac{\theta+1}{2}\left(\Theta_{+}^{k-1}-\Theta_{-}^{k-1}\right) & \frac{\theta-1}{2}\left(\Theta_{+}^{k-1}+\Theta_{-}^{k-1}\right)  \\
    \end{array}
  \right).
\end{equation}

Let us study some properties of this matrix $\mathbb{A}.$ The next
proposition deals with eigenvalues of the matrix $A$.

\begin{proposition}\label{A-N} Let $\mathbb{A}$ be the matrix given by \eqref{mainformofmainMatrix}. Then the following assertions hold true:
\begin{itemize}
  \item [(i)] The matrix $\mathbb{A}$ has the following form
\begin{eqnarray}\label{anotherformofA}
\mathbb{A}=\dfrac{1}{(\theta+t_3)(\theta t_3+1)}\left(
                                             \begin{array}{cc}
                                               \frac{\theta+1}{2}(t_3^2+2\theta t_3+1) & \frac{\theta-1}{2}(t_3^2-1) \\
                                               \frac{\theta+1}{2}(t_3^2-1) & \frac{\theta-1}{2}(t_3^2+2\theta t_3+1) \\
                                             \end{array}
                                           \right)
\end{eqnarray}
  where $t_3$ is a fixed point of the function given by \eqref{gthetaoft};
  \item [(ii)] The numbers $\lambda_1=1$, $\lambda_2=\det(\mathbb{A})\in (0,1)$ are eigenvalues of the matrix $A;$
  \item [(iii)] The vectors
  \begin{eqnarray}
  \label{eigenvectorlambda1}(x_1,y_1)&=&\left(t_3+1,t_3-1\right),\\
  \label{eigenvectorlambda2}(x_2,y_2)&=&\left(-(\theta-1)(t_3-1),(\theta+1)(t_3+1)\right)
  \end{eqnarray} are eigenvectors of the matrix $\mathbb{A}$ corresponding to the eigenvalues $\lambda_1=1$ and $\lambda_2=\det(A),$ respectively;
  \item [(iv)] If the matrix $P$ has the following form
             $$P=\left(
              \begin{array}{cc}
                t_3+1 & -(\theta-1)(t_3-1) \\
                t_3-1 & (\theta+1)(t_3+1) \\
              \end{array}
            \right),$$  then \begin{eqnarray}\label{diagonalformofmainMatrix}
            P^{-1}\mathbb{A}P=\left(\begin{array}{cc}
                                 \lambda_1 & 0 \\
                                  0 & \lambda_2 \\
                           \end{array}
                     \right);
            \end{eqnarray}
  \item [(v)] For any $n\in \bn$ one has
  \begin{eqnarray}
  \mathbb{A}^{n}=\left(
         \begin{array}{cc}
            \dfrac{(\theta+1)x_1^2+(\theta-1)y_1^2\lambda_2^n}{(\theta+1)x_1^2+(\theta-1)y_1^2} & \dfrac{x_1y_1(\theta-1)(1-\lambda_2^n)}{(\theta+1)x_1^2+(\theta-1)y_1^2} \\
            \dfrac{x_1y_1(\theta+1)(1-\lambda_2^n)}{(\theta+1)x_1^2+(\theta-1)y_1^2} & \dfrac{(\theta+1)x_1^2\lambda_2^n+(\theta-1)y_1^2}{(\theta+1)x_1^2+(\theta-1)y_1^2}\\
         \end{array}
    \right),
  \end{eqnarray}
  where $(x_1,y_1)$ is an eigenvector of the matrix $\mathbb{A}.$
\end{itemize}
\end{proposition}
\begin{proof} (i). We know that $t_3$ is a fixed point of \eqref{gthetaoft}, i.e.,
\begin{eqnarray*}
\frac{\theta\sqrt[k-1]{t_3}-1}{\theta-\sqrt[k-1]{t_3}}=t_3.
\end{eqnarray*}
It follows from the last identity that
\begin{eqnarray}\label{t_3theta}
t_3=\left(\dfrac{\theta t_3+1}{\theta+t_3}\right)^k.
\end{eqnarray}
By means \eqref{gamma03}, \eqref{Theta+Theta-}, \eqref{t_3theta} one
can easily get that the matrix $\mathbb{A}$ has the form
\eqref{anotherformofA}.

(ii). We know that the following equation
\begin{equation*}
\lambda^2-\tr(\mathbb{A})\lambda+\det(\mathbb{A})=0
\end{equation*}
is a characteristic equation of the matrix $\mathbb{A}$ given by
\eqref{anotherformofA}. Without forcing by detail we can make sure
\begin{eqnarray*}
\tr(\mathbb{A})-\det(\mathbb{A})=1,
\end{eqnarray*}
this means that $\lambda_1=1$ and
$\lambda_2=\det(\mathbb{A})=\frac{(\theta^2-1)t_3}{(\theta+t_3)(\theta
t_3+1)}$ are eigenvalues of the matrix $\mathbb{A}$.

(iii). The eigenvector $(x_1,y_1)$ of the matrix $\mathbb{A}$,
corresponding to $\lambda_1=1$ satisfies the following equation
\begin{eqnarray*}
(t_3-1)x_1=(t_3+1)y_1.
\end{eqnarray*}
Then, one finds
\begin{eqnarray*}\left\{\begin{array}{l}
x_1=t_3+1\\
y_1=t_3-1.
\end{array}
\right.
\end{eqnarray*}
Analogously, one can show that the eigenvector $(x_2,y_2)$ of the
matrix $\mathbb{A}$, corresponding to $\lambda_2=\det(\mathbb{A})$,
is equal to
\begin{eqnarray*}\left\{\begin{array}{l}
x_2=-(\theta-1)(t_3-1)\\
y_2= (\theta+1)(t_3+1).
\end{array}
\right.
\end{eqnarray*}
It is worth noting that $(x_2,y_2)=\left(-(\theta-1)y_1,(\theta+1)x_1\right).$

(iv). It is clear that
$$P=\left(
              \begin{array}{cc}
                x_1 & x_2 \\
                y_1 & y_2 \\
              \end{array}
            \right),
$$
where the vectors $(x_1,y_1)$ and $(x_2,y_2)$ are defined by \eqref{eigenvectorlambda1}, \eqref{eigenvectorlambda2}. We then get
\begin{eqnarray*}
P^{-1}\mathbb{A}P&=&\frac{1}{\det(P)}\left(
                              \begin{array}{cc}
                                y_2 & -x_2 \\
                                -y_1 & x_1 \\
                              \end{array}
                            \right)
\left(
       \begin{array}{cc}
         \lambda_1x_1 & \lambda_2x_2 \\
         \lambda_1y_1 & \lambda_2y_2 \\
       \end{array}
     \right)=\left(
               \begin{array}{cc}
                 \lambda_1 & 0 \\
                 0 & \lambda_2 \\
               \end{array}
             \right),
\end{eqnarray*}
where $\det(P)=(\theta+1)x_1^2+(\theta-1)y_1^2>0.$

(v). From \eqref{diagonalformofmainMatrix} it follows that
\[\mathbb{A}=P\left(
       \begin{array}{cc}
         \lambda_1 & 0 \\
         0 & \lambda_2 \\
       \end{array}
     \right)P^{-1}.
\]
Therefore, for any $n\in\bn$ we obtain
\begin{eqnarray*}
\mathbb{A}^n&=&P\left(
       \begin{array}{cc}
         \lambda_1^n & 0 \\
         0 & \lambda_2^n \\
       \end{array}
     \right)P^{-1}=\frac{1}{\det(P)}\left(
              \begin{array}{cc}
                x_1 & x_2 \\
                y_1 & y_2 \\
              \end{array}
            \right)\left(
                          \begin{array}{cc}
                             y_2\lambda_1^n & -x_2\lambda_1^n \\
                             -y_1\lambda_2^n & x_1\lambda_2^n \\
                          \end{array}
                     \right)\\
&=&\frac{1}{\det(P)}\left(
                           \begin{array}{cc}
                             x_1y_2\lambda_1^n-x_2y_1\lambda_2^n & x_1x_2(\lambda_2^n-\lambda_1^n) \\
                             y_1y_2(\lambda_1^n-\lambda_2^n) & x_1y_2\lambda_2^n-x_2y_1\lambda_1^n \\
                           \end{array}
                         \right)\\
&=&\left(
         \begin{array}{cc}
            \dfrac{(\theta+1)x_1^2+(\theta-1)y_1^2\lambda_2^n}{(\theta+1)x_1^2+(\theta-1)y_1^2} & \dfrac{x_1y_1(\theta-1)(1-\lambda_2^n)}{(\theta+1)x_1^2+(\theta-1)y_1^2} \\
            \dfrac{x_1y_1(\theta+1)(1-\lambda_2^n)}{(\theta+1)x_1^2+(\theta-1)y_1^2} & \dfrac{(\theta+1)x_1^2\lambda_2^n+(\theta-1)y_1^2}{(\theta+1)x_1^2+(\theta-1)y_1^2}\\
         \end{array}
    \right).
\end{eqnarray*}
This completes the proof.
\end{proof}

In what follows, for the sake of simplicity, let us denote
\begin{eqnarray}\label{numbersK_i}
K_0=\frac{\sqrt{\theta}+1}{2}, \quad K_3=\frac{\sqrt{\theta}-1}{2},
\end{eqnarray}
here as before $\theta=\exp\{2\beta\}$. In these notations, the
operator $K_{<u,v>}$ given by \eqref{K<u,v>} can be written as
follows
\begin{eqnarray}\label{compactformofK<u,v>}
K_{<u,v>}=K_{0}\sigma_0^{(u)}\otimes\sigma_0^{(v)}+K_{3}\sigma_3^{(u)}\otimes\sigma_3^{(v)}.
\end{eqnarray}
\begin{remark}\label{K02+K32} In the sequel,  we will frequently use the following identities for the numbers $K_0,$ $K_3$ given by \eqref{numbersK_i}:
\begin{itemize}
  \item [(i)] $K_0^2+K_3^2=\dfrac{\theta+1}{2};$
  \item [(ii)] $2K_0K_3=\dfrac{\theta-1}{2}.$
\end{itemize}
\end{remark}

\begin{proposition}\label{Traceh^i}
Let $K_{<u,v>}$ be given by \eqref{compactformofK<u,v>}, $\overrightarrow{S(x)}=(1,2,\cdots,k),$ and $\bh^{(i)}=h_0^{(i)}\sigma_0^{(i)}+h_3^{(i)}\sigma_3^{(i)},$ where $i\in \overrightarrow{S(x)}.$ Then we have
\begin{eqnarray}
\tr_{x]}\left[\prod_{i\in \overrightarrow{S(x)}}K_{<x,i>}\prod_{i\in \overrightarrow{S(x)}}\bh^{(i)}\prod_{i\in \overleftarrow{S(x)}}K_{<x,i>}\right]=g_0^{(x)}\sigma_0^{(x)}+g_3^{(x)}\sigma_3^{(x)}
\end{eqnarray}
where
\begin{eqnarray}
\left(g_0^{(x)},g_3^{(x)}\right)&=&A_{h^{(1)}}A_{h^{(2)}}\cdots A_{h^{(k-1)}}A_{h^{(k)}}e_1,\quad e_1=(1,0),\\
\label{A_h^i}A_{h^{(i)}}&=&\left(
                \begin{array}{cc}
                  (K_0^2+K_3^2)h_0^{(i)} & 2K_0K_3h_0^{(i)} \\
                  2K_0K_3h_3^{(i)} & (K_0^2+K_3^2)h_0^{(i)} \\
                \end{array}
              \right), \quad i\in \overrightarrow{S(x)}.
\end{eqnarray}
\end{proposition}
\begin{proof}
It is clear that
\begin{eqnarray*}
&&\tr_{x]}\left[\prod_{i\in \overrightarrow{S(x)}}K_{<x,i>}\prod_{i\in \overrightarrow{S(x)}}\bh^{(i)}\prod_{i\in \overleftarrow{S(x)}}K_{<x,i>}\right]=\\
&&=\tr_{x]}\left[K_{<x,1>}\bh^{(1)}\tr_{x]}\left[K_{<x,2>}\bh^{(2)}\cdots\tr_{x]}
\left[K_{<x,k>}\bh^{(k)}K_{<x,k>}\right]\cdots K_{<x,2>}\right]K_{<x,1>}\right]
\end{eqnarray*}

Let us first evaluate ${\bg}_k^{(x)}:=\tr_{x]}\left[K_{<x,k>}\bh^{(k)}K_{<x,k>}\right].$ From \eqref{compactformofK<u,v>} it follows that
\begin{eqnarray*}
K_{<x,k>}\bh^{(k)}K_{<x,k>}&=&K^2_0\sigma_{0}^{(x)}\otimes
\left(h_0^{(k)}\sigma_0^{(k)}+h_3^{(k)}\sigma_3^{(k)}\right)\\
&&+K_0K_3\sigma_{3}^{(x)}\otimes
\left(h_0^{(k)}\sigma_3^{(k)}+h_3^{(k)}\sigma_0^{(k)}\right)\\
&&+K_0K_3\sigma_{3}^{(x)}\otimes
\left(h_0^{(k)}\sigma_3^{(k)}+h_3^{(k)}\sigma_0^{(k)}\right)\\
&&+K^2_3\sigma_{0}^{(x)}\otimes
\left(h_0^{(k)}\sigma_0^{(k)}+h_3^{(k)}\sigma_3^{(k)}\right)
\end{eqnarray*}
Therefore, one gets
\begin{eqnarray}\label{g3ofTrh3}
{\bg}_k^{(x)}&=&g_{k,0}^{(x)}\sigma_{0}^{(x)}+g_{k,3}^{(x)}\sigma_{3}^{(x)}
\end{eqnarray}
where
\begin{eqnarray}\label{g_1^(3)}
(g_{k,0}^{(x)}, g_{k,3}^{(x)})&=&A_{h^{(k)}}e_1,\quad e_1=(1,0).
\end{eqnarray}

Now, evaluate ${\bg}_{k-1}^{(x)}:=\tr_{x]}\left[K_{<x,k-1>}\bh^{(k-1)}{\bg}_{k}^{(x)}K_{<x,k-1>}\right].$ Using \eqref{compactformofK<u,v>} and \eqref{g3ofTrh3} we find
\begin{eqnarray*}
K_{<x,k-1>}\bh^{(k-1)}{\bg}_{k}^{(x)}K_{<x,k-1>}&=&
K^2_0\left(g_{k,0}^{(x)}\sigma_{0}^{(x)}+g_{k,3}^{(x)}\sigma_{3}^{(x)}\right)\otimes
\left(h_0^{(k)}\sigma_0^{(k)}+h_3^{(k)}\sigma_3^{(k)}\right)\\
&&+K_0K_3\left(g_{k,0}^{(x)}\sigma_{3}^{(x)}+g_{k,3}^{(x)}\sigma_{0}^{(x)}\right)\otimes
\left(h_0^{(k)}\sigma_3^{(k)}+h_3^{(k)}\sigma_0^{(k)}\right)\\
&&+K_0K_3\left(g_{k,0}^{(x)}\sigma_{3}^{(x)}+g_{k,3}^{(x)}\sigma_{0}^{(x)}\right)\otimes
\left(h_0^{(k)}\sigma_3^{(k)}+h_3^{(k)}\sigma_0^{(k)}\right)\\
&&+K^2_3\left(g_{k,0}^{(x)}\sigma_{0}^{(x)}+g_{k,3}^{(x)}\sigma_{3}^{(x)}\right)\otimes
\left(h_0^{(k)}\sigma_0^{(k)}+h_3^{(k)}\sigma_3^{(k)}\right)
\end{eqnarray*}
Hence, one has
\begin{eqnarray}
{\bg}_{k-1}^{(x)}=g_{k-1,0}^{(x)}\sigma_{0}^{(x)}+g_{k-1,3}^{(x)}\sigma_{3}^{(x)}
\end{eqnarray}
where
\begin{eqnarray}
\left(g_{k-1,0}^{(x)},g_{k-1,3}^{(x)}\right)=A_{h^{(k-1)}}\left(g_{k,0}^{(x)},g_{k,3}^{(x)}\right)=A_{h^{(k-1)}}A_{h^{(k)}}e_1
\end{eqnarray}

Similarly, one can evaluate
\begin{eqnarray}
{\bg}_i^{(x)}:=\tr_{x]}\left[K_{<x,i>}\bh^{(i)}{\bg}_{i+1}^{(x)}K_{<x,i>}\right]=g_{i,0}^{(x)}\sigma_{0}^{(x)}+g_{i,3}^{(x)}\sigma_{3}^{(x)}
\end{eqnarray}
where
\begin{eqnarray}
\left(g_{i,0}^{(x)},g_{i,3}^{(x)}\right)=A_{h^{(i)}}\left(g_{i+1,0}^{(x)},g_{i+1,3}^{(x)}\right)=\cdots
=A_{h^{(i)}}A_{h^{(i+1)}}\cdots A_{h^{(k-1)}}A_{h^{(k)}}e_1
\end{eqnarray}
Consequently, we have
\begin{eqnarray*}
\tr_{x]}\left[\prod_{i\in \overrightarrow{S(x)}}K_{<x,i>}\prod_{i\in \overrightarrow{S(x)}}\bh^{(i)}\prod_{i\in \overleftarrow{S(x)}}K_{<x,i>}\right]=g_0^{(x)}\sigma_0^{(x)}+g_3^{(x)}\sigma_3^{(x)},
\end{eqnarray*}
where
\begin{eqnarray*}
\left(g_0^{(x)},g_3^{(x)}\right)=A_{h^{(1)}}A_{h^{(2)}}\cdots A_{h^{(k-1)}}A_{h^{(k)}}e_1.
\end{eqnarray*}
This completes the proof.
\end{proof}
\begin{remark}
One can easily check that for any permutation $\pi$ of the set $\overrightarrow{S(x)}$ we have
$$A_{h^{(1)}}A_{h^{(2)}}\cdots A_{h^{(k-1)}}A_{h^{(k)}}=A_{h^{(\pi(1))}}A_{h^{(\pi(2))}}\cdots A_{h^{(\pi(k-1))}}A_{h^{(\pi(k))}},$$
in other words the matrices $A_{h^{(i)}},$ $i\in \overrightarrow{S(x)}$ commute each other.
\end{remark}

\begin{corollary}\label{traceforalpha_0}
Let $K_{<u,v>}$ be  given by \eqref{compactformofK<u,v>}, $\overrightarrow{S(x)}=(1,2,\cdots,k),$ and
\begin{eqnarray*}
\bh^{(1)}=h_3\sigma_3^{(1)}, \quad
\bh^{(i)}=\alpha_0\sigma_0^{(i)}, \quad i={\overline{2,k}},
\end{eqnarray*}
where $\a_0=\sqrt[k-1]{\Theta^{k}},$ $\Theta=\dfrac{2}{\theta+1},$ and $h_3$ is some positive number.
Then we have
\begin{eqnarray}
\tr_{x]}\left[\prod_{i\in \overrightarrow{S(x)}}K_{<x,i>}\prod_{i\in \overrightarrow{S(x)}}\bh^{(i)}\prod_{i\in \overleftarrow{S(x)}}K_{<x,i>}\right]=h_3\frac{\theta-1}{\theta+1}\sigma_3^{(x)}
\end{eqnarray}
\end{corollary}
\begin{proof} Let us calculate the matrices $A_{h^{(i)}},$ $i\in\overrightarrow{S(x)}$ which are given by \eqref{A_h^i}.
It is clear that
\begin{eqnarray*}
A_{h^{(1)}}=\left(
  \begin{array}{cc}
    0 & 2K_0K_3h_3 \\
    2K_0K_3h_3 & 0 \\
  \end{array}
\right),\quad A_{h^{(i)}}=\left(
  \begin{array}{cc}
    (K_0^2+K_3^2)\a_0 & 0 \\
    0 & (K_0^2+K_3^2)\a_0 \\
  \end{array}
\right),
\end{eqnarray*}
where $\ i=\overline{2,k}.$ We then have
\begin{eqnarray*}
A_{h^{(1)}}A_{h^{(2)}}\cdots A_{h^{(k)}}=\left(
                               \begin{array}{cc}
                                 0 & 2\a_0^{k-1}(K_0^2+K_3^2)^{k-1}K_0K_3h_3 \\
                                 2\a_0^{k-1}(K_0^2+K_3^2)^{k-1}K_0K_3h_3 & 0 \\
                               \end{array}
                             \right).
\end{eqnarray*}
Therefore, it follows from Remark \ref{K02+K32} and Proposition \ref{Traceh^i} that
\begin{eqnarray*}
\tr_{x]}\left[\prod_{i\in \overrightarrow{S(x)}}K_{<x,i>}\prod_{i\in \overrightarrow{S(x)}}\bh^{(i)}\prod_{i\in \overleftarrow{S(x)}}K_{<x,i>}\right]&=&h_3\a_0^{k-1}\left(\frac{\theta-1}{2}\right)\left(\frac{\theta+1}{2}\right)^{k-1}\sigma_3^{(x)}\\
&=&h_3\frac{\theta-1}{\theta+1}\sigma_3^{(x)}
\end{eqnarray*}
\end{proof}
\begin{corollary}\label{traceforgamma_0gamma_1}
Let $K_{<u,v>}$ be  given by \eqref{compactformofK<u,v>}, $\overrightarrow{S(x)}=(1,2,\cdots,k),$ and
\begin{eqnarray*}
\bh^{(1)}=h_0\sigma_0^{(1)}+h_3\sigma_3^{(1)},\quad \bh^{(i)}=\gamma_0\sigma_0^{(i)}+\gamma_3\sigma_3^{(i)},\ i=\overline{2,k},
\end{eqnarray*}
where the numbers $\gamma_0,\gamma_3$ are given by \eqref{gamma03} and $h_0,h_3$ are some positive numbers.
Then we have
\begin{eqnarray}
\tr_{x]}\left[\prod_{i\in \overrightarrow{S(x)}}K_{<x,i>}\prod_{i\in \overrightarrow{S(x)}}\bh^{(i)}\prod_{i\in \overleftarrow{S(x)}}K_{<x,i>}\right]=h_{0}^{(x)}\sigma_0^{(x)}+h_{3}^{(x)}\sigma_3^{(x)},
\end{eqnarray}
where $h^{(x)}=\mathbb{A}h$ and $\mathbb{A}$ is a matrix given by
\eqref{mainformofmainMatrix}, $h^{(x)}=(h_0^{(x)},h_3^{(x)}),$
$h=(h_0,h_3)$ are vectors.
\end{corollary}
\begin{proof}
Let us calculate the matrices $A_{h^{(i)}},$ $i\in\overrightarrow{S(x)}$ which are given by \eqref{A_h^i}. It follows from Remark \ref{K02+K32} that
\begin{eqnarray*}
A_{h^{(1)}}=\left(
  \begin{array}{cc}
    h_0\dfrac{\theta+1}{2} & h_3\dfrac{\theta-1}{2} \\
    h_3\dfrac{\theta-1}{2} & h_0\dfrac{\theta+1}{2} \\
  \end{array}
\right),\quad A_{h^{(i)}}=\left(
  \begin{array}{cc}
    \gamma_0\dfrac{\theta+1}{2} & \gamma_3\dfrac{\theta-1}{2} \\
    \gamma_3\dfrac{\theta-1}{2} & \gamma_0\dfrac{\theta+1}{2} \\
  \end{array}
\right),
\end{eqnarray*}
where $\ i=\overline{2,k}.$ By means \eqref{M^n}  we then get
\begin{eqnarray*}
A_{h^{(2)}}\cdots A_{h^{(k)}}=A_{h^{(2)}}^{k-1}=\frac{1}{2}\left(
                               \begin{array}{cc}
                                 \Theta_{+}^{k-1}+\Theta_{-}^{k-1} & \Theta_{+}^{k-1}-\Theta_{-}^{k-1} \\
                                 \Theta_{+}^{k-1}-\Theta_{-}^{k-1} & \Theta_{+}^{k-1}+\Theta_{-}^{k-1}  \\
                               \end{array}
                             \right).
\end{eqnarray*}
here as before $\Theta_{+},\Theta_{-}$ are given by
\eqref{Theta+Theta-}. After simple algebra, it follows from
Proposition \ref{Traceh^i} that
\begin{eqnarray*}
h^{(x)}=A_{h^{(1)}}A_{h^{(2)}}\cdots A_{h^{(k)}}e_1=\mathbb{A}h
\end{eqnarray*}
where the matrix $\mathbb{A}$ is given by
\eqref{mainformofmainMatrix} and $h^{(x)}=(h_0^{(x)},h_3^{(x)}),$
$h=(h_0,h_3)$ are vectors.
\end{proof}
Let us consider the following elements:
\begin{eqnarray}
\sigma_0^{\Lambda}:=\bigotimes_{x\in \Lambda}\sigma_0^{(x)}\in\cb_{\Lambda}, \ \Lambda\subset\Lambda_n,\quad\quad
\sigma_3^{\overrightarrow{S(x)},1}:=\sigma_3^{(1)}\otimes\sigma_0^{(2)}\otimes\cdots\otimes\sigma_0^{(k)}\in \cb_{S(x)},\\
\sigma_3^{\overrightarrow{W}_{n+1},1}:=\sigma_3^{\overrightarrow{S(x_{W_n}^{(1)})},1}
\otimes\sigma_0^{\overrightarrow{W}_{n+1}\setminus\overrightarrow{S(x_{W_n}^{(1)})}}\in\cb_{W_{n+1}},\\
\label{elementa}a_{\sigma_{3}}^{\Lambda_{n+1}}:=\bigotimes_{i=0}^n\sigma_0^{\overrightarrow{W}_{i}}\otimes\sigma_3^{\overrightarrow{W}_{n+1},1}
\in\cb_{\Lambda_{n+1}}.
\end{eqnarray}

\begin{proposition}\label{evaluatestateahlpha}
Let $\ffi^{(f)}_{w_0(\a_0),{\bh(\alpha_0)}}$ be a forward QMC
corresponding to the model \eqref{compactformofK<u,v>} with boundary
conditions $\omega_0(\a_0)=\frac{1}{\a_0}\sigma_0$ and
$\bh^{(x)}=\alpha_0\sigma_0^{(x)}$ for all $x\in L,$ where
$\alpha_0=\sqrt[k-1]{\Theta^k},$ $\Theta=\frac{2}{\theta+1}.$ Let
$a_{\sigma_{3}}^{\Lambda_{N+1}}$ be an element given by
\eqref{elementa} and  $\theta>\frac{k+1}{k-1}.$ Then one has
$\ffi^{(f)}_{w_0(\a_0),{\bh(\alpha_0)}}\left(a_{\sigma_{3}}^{\Lambda_{N+1}}\right)=0,$
for any $N\in\bn.$
\end{proposition}

\begin{proof}
Due to \eqref{eq2} (see Theorem \ref{compa}) the compatibility condition holds
$\ffi^{(n+1,f)}_{w_0(\a_0),\bh(\alpha_0)}\lceil_{\cb_{\L_n}}=\ffi^{(n,f)}_{w_0(\a_0),\bh(\alpha_0)}.$ Therefore,
\begin{eqnarray}
\ffi^{(f)}_{w_0(\a_0),{\bh(\alpha_0)}}\left(a_{\sigma_{3}}^{\Lambda_{N+1}}\right)=
w-\lim_{n\to\infty}\ffi^{(n,f)}_{w_0(\a_0),\bh(\alpha_0)}
\left(a_{\sigma_{3}}^{\Lambda_{N+1}}\right)=\ffi^{(N+1,f)}_{w_0(\a_0),\bh(\alpha_0)}\left(a_{\sigma_{3}}^{\Lambda_{N+1}}\right).
\end{eqnarray}
Taking into account $w_0(\a_0)=\frac{1}{\alpha_0}\sigma_0$ and due to Proposition \ref{state^nwithW_n}, it is enough to evaluate the following
\begin{eqnarray}
\ffi^{(N+1,f)}_{w_0(\a_0),\bh(\alpha_0)}\left(a_{\sigma_{3}}^{\Lambda_{N+1}}\right)
&=&\tr\left(\cw_{N+1]}\left(a_{\sigma_{3}}^{\Lambda_{N+1}}\right)\right)\nonumber\\
&=&\frac{1}{\alpha_0}\tr\left[K_{[0,1]}\cdots K_{[N,N+1]}\bh_{N+1}K^{*}_{[N,N+1]}\cdots K^{*}_{[0,1]}
a_{\sigma_{3}}^{\Lambda_{N+1}}\right]\nonumber\\
&=&\frac{1}{\alpha_0}\tr\Bigl[K_{[0,1]}\cdots K_{[N-1,N]}\Bigr.\nonumber\\
&&\left.\quad\quad\quad\quad\tr_{N]}\left[K_{[N,N+1]}\bh_{N+1}K^{*}_{[N,N+1]}\sigma_3^{\overrightarrow{W}_{N+1},1}\right]K^{*}_{[N-1,N]}\cdots K^{*}_{[0,1]}\right].\nonumber
\end{eqnarray}
Now let us calculate $\widetilde{\bh}_N:=\tr_{N]}\left[K_{[N,N+1]}\bh_{N+1}K^{*}_{[N,N+1]}\sigma_3^{\overrightarrow{W}_{n+1},1}\right].$ Since $K_{<u,v>}$ is a self-adjoint, we then get
\begin{eqnarray*}
{\widetilde{\bh}}_N&=&\tr_{\left.x^{(1)}_{W_{N}}\right]}\left[\prod_{y\in \overrightarrow{S(x^{(1)}_{W_{N}})}}K_{\left\langle x^{(1)}_{W_{N}},y\right\rangle}\prod_{y\in \overrightarrow{S(x^{(1)}_{W_{N}})}}\bh^{(y)}\prod_{y\in \overleftarrow{S(x^{(1)}_{W_{N}})}}K_{\left\langle x^{(1)}_{W_{N}},y\right\rangle}\sigma_3^{\overrightarrow{S(x_{W_N}^{(1)})},1}\right]\otimes\\
&&\bigotimes_{x\in \overrightarrow{W}_{N}\setminus x^{(1)}_{W_{N}}}\tr_{x]}\left[\prod_{y\in \overrightarrow{S(x)}}K_{<x,y>}\prod_{y\in \overrightarrow{S(x)}}\bh^{(y)}\prod_{y\in \overleftarrow{S(x)}}K_{<x,y>}\right].
\end{eqnarray*}
We know that
\begin{eqnarray}\label{hhhhhh}
\tr_{x]}\left[\prod_{y\in \overrightarrow{S(x)}}K_{<x,y>}\prod_{y\in \overrightarrow{S(x)}}\bh^{(y)}\prod_{y\in \overleftarrow{S(x)}}K_{<x,y>}\right]=\bh^{(x)},
\end{eqnarray}
for every $x\in \overrightarrow{W}_{N}\setminus x^{(1)}_{W_{N}}.$ On the other hand, since operators $K_{<u,v>}$ and $\sigma_3^{(x)}$ commute each other for any $u,v,x\in L$ it follows from Corollary \ref{traceforalpha_0} that
\begin{eqnarray}
\tr_{\left.x^{(1)}_{W_{N}}\right]}\left[\prod_{y\in \overrightarrow{S(x^{(1)}_{W_{N}})}}K_{\left\langle x^{(1)}_{W_{N}},y\right\rangle}\prod_{y\in \overrightarrow{S(x^{(1)}_{W_{N}})}}\bh^{(y)}\prod_{y\in \overleftarrow{S(x^{(1)}_{W_{N}})}}K_{\left\langle x^{(1)}_{W_{N}},y\right\rangle}\sigma_3^{\overrightarrow{S(x_{W_N}^{(1)})},1}\right]={\widetilde{\bh}}^{(x^{(1)}_{W_{N}})},
\end{eqnarray}
where
\begin{eqnarray*}
{\widetilde{\bh}}^{(x^{(1)}_{W_{N}})}=\alpha_0\frac{\theta-1}{\theta+1}\sigma_3^{(x^{(1)}_{W_{N}})}.
\end{eqnarray*}
Hence, we obtain
\begin{eqnarray*}
{\widetilde{\bh}}_N={\widetilde{\bh}}^{(x^{(1)}_{W_{N}})}\bigotimes_{x\in \overrightarrow{W}_{N}\setminus x^{(1)}_{W_{N}}}\bh^{(x)}.
\end{eqnarray*} Therefore, one finds
\begin{eqnarray*}
\ffi^{(N+1,f)}_{w_0,\bh(\alpha_0)}\left(a_{\sigma_{3}}^{\Lambda_{N+1}}\right)&=&\frac{1}{\alpha_0}\tr\Bigl[K_{[0,1]}\cdots K_{[N-2,N-1]}\Bigr.\nonumber\\
&&\left.\quad\quad\quad\quad\quad\tr_{N-1]}\left[K_{[N-1,N]}{\widetilde{\bh}}_{N}K^{*}_{[N-1,N]}\right]K^{*}_{[N-2,N-1]}\cdots K^{*}_{[0,1]}\right].
\end{eqnarray*}
So, after $N$ times applying Corollary \eqref{traceforalpha_0}, we get
\begin{eqnarray*}
\ffi^{(N+1,f)}_{w_0,\bh(\alpha_0)}\left(a_{\sigma_{3}}^{\Lambda_{N+1}}\right)
=\a_0^{N-1}\left(\frac{\theta-1}{\theta+1}\right)^{2N}\tr(\sigma_3^{(0)})=0.
\end{eqnarray*}
This completes the proof.
\end{proof}

\begin{proposition}\label{evaluatestategamma}
Let $\ffi^{(f)}_{w_0(\gamma),{\bh(\gamma)}}$ be a forward QMC
corresponding to the model \eqref{compactformofK<u,v>} with boundary
conditions $\omega_0(\gamma)=\frac{1}{\gamma_0}\sigma_0$ and
$\bh^{(x)}=\gamma_0\sigma_0^{(x)}+\gamma_3\sigma_3^{(x)}$ for all
$x\in L,$ where $\gamma_0, \gamma_3$ are given by \eqref{gamma03}.
Let $a_{\sigma_{3}}^{\Lambda_{N+1}}$ be an element given by
\eqref{elementa} and  $\theta>\frac{k+1}{k-1}.$ Then one has
\begin{eqnarray}
\ffi^{(f)}_{w_0,{\bh(\gamma)}}\left(a_{\sigma_{3}}^{\Lambda_{N+1}}\right)=
\frac{1}{\gamma_0}\Bigl\langle
\mathbb{A}^{N+1}h_{\gamma_0,\gamma_3},e\Bigr\rangle\quad \quad
\forall N\in\bn,
\end{eqnarray}
where $\mathbb{A}$ is a matrix given by \eqref{anotherformofA},
$\Bigl\langle\cdot,\cdot\Bigr\rangle$ is a inner product of vectors
and $e=(1,0),$ $h_{\gamma_0,\gamma_3}=(\gamma_3,\gamma_0)$ are
vectors.
\end{proposition}
\begin{proof}
Again the compatibility condition yields that
\begin{eqnarray}
\ffi^{(f)}_{w_0,{\bh(\gamma)}}\left(a_{\sigma_{3}}^{\Lambda_{N+1}}\right)=w-\lim_{n\to\infty}\ffi^{(n,f)}_{w_0,\bh(\gamma)}
\left(a_{\sigma_{3}}^{\Lambda_{N+1}}\right)=\ffi^{(N+1,f)}_{w_0,\bh(\gamma)}\left(a_{\sigma_{3}}^{\Lambda_{N+1}}\right).
\end{eqnarray}
Due to Proposition \ref{state^nwithW_n}, it is enough to evaluate the following
\begin{eqnarray}
\ffi^{(N+1,f)}_{w_0,\bh(\gamma)}\left(a_{\sigma_{3}}^{\Lambda_{N+1}}\right)&=&\frac{1}{\gamma_0}\tr\Bigl[K_{[0,1]}\cdots K_{[N-1,N]}\Bigr.\nonumber\\
&&\left.\quad\quad\quad\tr_{N]}\left[K_{[N,N+1]}\bh_{N+1}K^{*}_{[N,N+1]}\sigma_3^{\overrightarrow{W}_{N+1},1}\right]K^{*}_{[N-1,N]}\cdots K^{*}_{[0,1]}\right].\nonumber
\end{eqnarray}
Let us calculate $\widetilde{\bh}_N:=\tr_{N]}\left[K_{[N,N+1]}\bh_{N+1}K^{*}_{[N,N+1]}\sigma_3^{\overrightarrow{W}_{n+1},1}\right].$ Self-adjointness of $K_{<u,v>}$ implies that
\begin{eqnarray*}
{\widetilde{\bh}}_N&=&\tr_{\left.x^{(1)}_{W_{N}}\right]}\left[\prod_{y\in \overrightarrow{S(x^{(1)}_{W_{N}})}}K_{\left\langle x^{(1)}_{W_{N}},y\right\rangle}\prod_{y\in \overrightarrow{S(x^{(1)}_{W_{N}})}}\bh^{(y)}\prod_{y\in \overleftarrow{S(x^{(1)}_{W_{N}})}}K_{\left\langle x^{(1)}_{W_{N}},y\right\rangle}\sigma_3^{\overrightarrow{S(x_{W_N}^{(1)})},1}\right]\otimes\\
&&\bigotimes_{x\in \overrightarrow{W}_{N}\setminus x^{(1)}_{W_{N}}}\tr_{x]}\left[\prod_{y\in \overrightarrow{S(x)}}K_{<x,y>}\prod_{y\in \overrightarrow{S(x)}}\bh^{(y)}\prod_{y\in \overleftarrow{S(x)}}K_{<x,y>}\right].
\end{eqnarray*}
We know that
\begin{eqnarray*}
\tr_{x]}\left[\prod_{y\in \overrightarrow{S(x)}}K_{<x,y>}\prod_{y\in \overrightarrow{S(x)}}\bh^{(y)}\prod_{y\in \overleftarrow{S(x)}}K_{<x,y>}\right]=\bh^{(x)},
\end{eqnarray*}
for every $x\in \overrightarrow{W}_{N}\setminus x^{(1)}_{W_{N}}.$
On the other hand, since operators $K_{<u,v>}$ and $\sigma_3^{(x)}$ commute each other for any $u,v,x\in L$ it follows from Corollary \ref{traceforgamma_0gamma_1} that
\begin{eqnarray*}
\tr_{\left.x^{(1)}_{W_{N}}\right]}\left[\prod_{y\in \overrightarrow{S(x^{(1)}_{W_{N}})}}K_{\left\langle x^{(1)}_{W_{N}},y\right\rangle}\prod_{y\in \overrightarrow{S(x^{(1)}_{W_{N}})}}\bh^{(y)}\prod_{y\in \overleftarrow{S(x^{(1)}_{W_{N}})}}K_{\left\langle x^{(1)}_{W_{N}},y\right\rangle}\sigma_3^{\overrightarrow{S(x_{W_N}^{(1)})},1}\right]
={\widetilde{\bh}}^{(x^{(1)}_{W_{N}})},
\end{eqnarray*}
where
\begin{eqnarray*}
{\widetilde{\bh}}^{(x^{(1)}_{W_{N}})}=h_0\sigma_0^{(x^{(1)}_{W_{N}})}+h_3\sigma_3^{(x^{(1)}_{W_{N}})}, \ (h_0,h_3)=Ah_{\gamma_0,\gamma_3}, \ h_{\gamma_0,\gamma_3}=(\gamma_3,\gamma_0).
\end{eqnarray*}
 Thus we obtain
\begin{eqnarray*}
{\widetilde{\bh}}_N={\widetilde{\bh}}^{(x^{(1)}_{W_{N}})}\bigotimes_{x\in \overrightarrow{W}_{N}\setminus x^{(1)}_{W_{N}}}\bh^{(x)}.
\end{eqnarray*} Therefore, one gets
\begin{eqnarray*}
\ffi^{(N+1,f)}_{w_0,\bh(\gamma)}\left(a_{\sigma_{3}}^{\Lambda_{N+1}}\right)&=&\frac{1}{\gamma_0}\tr\Bigl[K_{[0,1]}\cdots K_{[N-2,N-1]}\Bigr.\nonumber\\
&&\left.\quad\quad\quad\quad\quad\tr_{N-1]}\left[K_{[N-1,N]}{\widetilde{\bh}}_{N}K^{*}_{[N-1,N]}\right]K^{*}_{[N-2,N-1]}\cdots K^{*}_{[0,1]}\right].
\end{eqnarray*}
Again applying $N$ times Corollary \ref{traceforgamma_0gamma_1}, one finds
\begin{eqnarray*}
\ffi^{(N+1,f)}_{w_0,\bh(\gamma)}\left(a_{\sigma_{3}}^{\Lambda_{N+1}}\right)
&=&\frac{1}{\gamma_0}\Bigl\langle
\mathbb{A}^{N+1}h_{\gamma_0,\gamma_3}, e\Bigr\rangle.
\end{eqnarray*}
This completes the proof.
\end{proof}

To prove our main result we are going to use the following theorem
(see \cite{BR}, Corollary 2.6.11).
\begin{theorem}\label{br-q}
Let $\varphi_1,$ $\varphi_2$ be two states on a quasi-local algebra $\ga=\cup_{\Lambda}\ga_\Lambda$. The states $\varphi_1,$ $\varphi_2$ are  quasi-equivalent if and only if for any given $\varepsilon>0$ there exists a finite volume $\Lambda\subset L$ such that $\|\varphi_1(a)-\varphi_2(a)\|<\varepsilon \|a\|$ for all $a\in B_{\Lambda^{'}}$ with $\Lambda^{'}\cap\Lambda=\emptyset.$
\end{theorem}

Now by means of the Theorem \ref{br-q} we will show that the states
$\ffi^{(f)}_{w_0(\a_0),{\bh(\alpha_0)}}$ and
$\ffi^{(f)}_{w_0(\gamma),{\bh(\gamma)}}$ are not quasi-equivalent.
Namely, we have the following

\begin{theorem}\label{alpha0gamma03}
Let $\theta>\frac{k+1}{k-1}$ and
$\ffi^{(f)}_{w_0(\a_0),{\bh(\alpha_0)}},$
$\ffi^{(f)}_{w_0(\gamma),{\bh(\gamma)}}$ be two forward QMC
corresponding to the model \eqref{compactformofK<u,v>} with two
boundary conditions $\omega_0(\a_0)=\frac{1}{\a_0}\sigma_0,$
$\bh^{(x)}=\alpha_0\sigma_0^{(x)},$  $\forall x\in L$ and
$\omega_0(\gamma)=\frac{1}{\gamma_0}\sigma_0,$
$\bh^{(x)}=\gamma_0\sigma_0^{(x)}+\gamma_3\sigma_3^{(x)},$  $\forall
x\in L,$ respectively, here as before
$\alpha_0=\sqrt[k-1]{\Theta^k},$ $\gamma_0$ and $\gamma_3$ are given
by \eqref{gamma03}. Then $\ffi^{(f)}_{w_0(\a_0),{\bh(\alpha_0)}}$
and $\ffi^{(f)}_{w_0(\gamma),{\bh(\gamma)}}$ are not
quasi-equivalent.
\end{theorem}
\begin{proof}
Let $a_{\sigma_{3}}^{\Lambda_{N+1}}$ be an
 element given by
\eqref{elementa}. It is clear that
$\left\|a_{\sigma_{3}}^{\Lambda_{N+1}}\right\|=1,$ for all
$N\in\bn.$

If $\theta>\frac{k+1}{k-1}$, then according to Propositions \ref{evaluatestateahlpha} and \ref{evaluatestategamma}, we have
\begin{eqnarray}
\label{valueofstatealpha}\ffi^{(f)}_{w_0(\a_0),{\bh(\alpha_0)}}\left(a_{\sigma_{3}}^{\Lambda_{N+1}}\right)&=&0,\\
\label{valueofstategamma0}
\ffi^{(f)}_{w_0(\gamma),{\bh(\gamma)}}\left(a_{\sigma_{3}}^{\Lambda_{N+1}}\right)&=&
\frac{1}{\gamma_0}\Bigl\langle
\mathbb{A}^{N+1}h_{\gamma_0,\gamma_3},e\Bigr\rangle
\end{eqnarray}
for all $N\in\bn,$ here as before $e=(1,0),$
$h_{\gamma_0,\gamma_3}=(\gamma_3,\gamma_0)$  and $\mathbb{A}$ is
given by \eqref{anotherformofA}. Then from
\eqref{valueofstategamma0} with Proposition \ref{A-N} one finds
\begin{eqnarray}\label{valueofstategamma}
\ffi^{(f)}_{w_0(\gamma),{\bh(\gamma)}}\left(a_{\sigma_{3}}^{\Lambda_{N+1}}\right)=\frac{(\theta+1)x_1^2\gamma_3+(\theta-1)x_1y_1\gamma_0}
{\gamma_0((\theta+1)x_1^2+(\theta-1)y_1^2)}+\frac{(\theta-1)(y_1^2\gamma_3-x_1y_1\gamma_0)}{\gamma_0((\theta+1)x_1^2+(\theta-1)y_1^2)}
\lambda_2^{N+1}.
\end{eqnarray}
where $\lambda_2$ is an eigenvalue of $\mathbb{A}$ and $(x_1,y_1)$
is an eigenvector of the matrix $\mathbb{A}$ corresponding to the
eigenvalue $\lambda_1=1$ (see Proposition \ref{A-N}). Since
$0<\lambda_2<1$ then there exists $N_0\in \bn$ such that
\begin{eqnarray}\label{estimateofvarepsilon}
\left|\frac{(\theta+1)x_1^2\gamma_3+(\theta-1)x_1y_1\gamma_0}
{\gamma_0((\theta+1)x_1^2+(\theta-1)y_1^2)}+\frac{(\theta-1)(y_1^2\gamma_3-x_1y_1\gamma_0)}{\gamma_0((\theta+1)x_1^2+(\theta-1)y_1^2)}
\lambda_2^{N+1}\right|\geq\nonumber\\
\ge \frac{(\theta+1)x_1^2\gamma_3+(\theta-1)x_1y_1\gamma_0}
{2\gamma_0((\theta+1)x_1^2+(\theta-1)y_1^2)}
\end{eqnarray}
for all $N>N_0.$

Now putting $\varepsilon_0=\cfrac{(\theta+1)x_1^2\gamma_3+(\theta-1)x_1y_1\gamma_0}
{2\gamma_0((\theta+1)x_1^2+(\theta-1)y_1^2)}$ and using \eqref{valueofstatealpha}, \eqref{valueofstategamma}, \eqref{estimateofvarepsilon} we obtain
\begin{eqnarray*}
\left|\ffi^{(f)}_{w_0(\a_0),{\bh(\alpha_0)}}\left(a_{\sigma_{3}}^{\Lambda_{N+1}}\right)
-\ffi^{(f)}_{w_0(\gamma),{\bh(\gamma)}}\left(a_{\sigma_{3}}^{\Lambda_{N+1}}\right)\right|\ge \varepsilon_0\left\|a_{\sigma_{3}}^{\Lambda_{N+1}}\right\|,
\end{eqnarray*}
for all $N>N_0,$ which means $\ffi^{(f)}_{w_0(\a_0),{\bh(\alpha_0)}}$ and $\ffi^{(f)}_{w_0(\gamma),{\bh(\gamma)}}$ are not quasi-equivalent. This completes the proof.
\end{proof}

Analogously, one can prove the following result.

\begin{theorem}\label{alpha0beta03}
Let $\theta>\frac{k+1}{k-1}$ and
$\ffi^{(f)}_{w_0(\a_0),{\bh(\alpha_0)}},$
$\ffi^{(f)}_{w_0(\beta),{\bh(\beta)}}$ be two forward quantum
$d-$Markov chains corresponding to the model
\eqref{compactformofK<u,v>} with two boundary conditions
$\omega_0(\a_0)=\frac{1}{\a_0}\sigma_0,$ $\bh^{(x)}=\alpha_0\sigma_0^{(x)},$  $\forall x\in L$ and
$\omega_0(\beta)=\frac{1}{\beta_0}\sigma_0,$ $\bh^{(x)}=\beta_0\sigma_0^{(x)}+\beta_3\sigma_3^{(x)},$  $\forall
x\in L,$ respectively, here as before
$\alpha_0=\sqrt[k-1]{\Theta^k},$ $\beta_0$ and
$\beta_3$ are given by \eqref{beta03}. Then $\ffi^{(f)}_{w_0(\a_0),{\bh(\alpha_0)}}$ and
$\ffi^{(f)}_{w_0(\beta),{\bh(\beta)}}$ are not quasi-equivalent.
\end{theorem}

From  Theorem \ref{alpha0gamma03} we immediately get the occurrence of
the phase transition for the model \eqref{compactformofK<u,v>} on the
Cayley tree of order $k$ in the regime $\theta>\cfrac{k+1}{k-1}$. This
completely proves our main Theorem \ref{main}.

%
%


\section*{Acknowledgement} The present study have been done within
the grant ERGS13-024-0057 of Malaysian Ministry of Higher Education.
The authors (F.M. and M.S) would like to thanks to the Abdus Salam
International Centre for Theoretical Physics, Trieste, Italy for
offering a Junior Associate Scheme fellowship.


\end{document}